%% file: main.tex
\documentclass{llncs}
\title{Wagner's Algorithm Provably Runs in Subexponential Time for $\SISinfty$}
\author{L\'{e}o Ducas\inst{1,2} \and Lynn Engelberts\inst{1,3} \and Johanna Loyer\inst{1}}
\institute{Centrum Wiskunde \& Informatica, The Netherlands \and Leiden University, The Netherlands \and QuSoft, The Netherlands} 

\input{preamble}

\DeleteAllComments

\newtoggle{bibtex}
\toggletrue{bibtex}
\iftoggle{bibtex}{ 
}{
	\usepackage[backend=biber, style=alphabetic, maxbibnames=99, maxalphanames=5, url=false]{biblatex}
    \addbibresource{references.bib}
}

\begin{document}

\maketitle
\thispagestyle{plain} % To show first page number 

\begin{abstract}
    At CRYPTO 2015, Kirchner and Fouque claimed that a carefully tuned variant of the Blum-Kalai-Wasserman (BKW) algorithm (JACM 2003) should solve the Learning with Errors problem (LWE) in slightly subexponential time for modulus $q=\mathrm{poly}(n)$ and narrow error distribution, when given enough LWE samples. 
    Taking a modular view, one may regard BKW as a combination of Wagner's algorithm (CRYPTO 2002), run over the corresponding dual problem, and the Aharonov-Regev distinguisher (JACM 2005). Hence the subexponential Wagner step alone should be of interest for solving this dual problem -- namely, the Short Integer Solution problem (SIS) -- but this appears to be undocumented so far.

    We re-interpret this Wagner step as walking backward through a chain of projected lattices, zigzagging through some auxiliary superlattices. We further randomize the bucketing step using Gaussian randomized rounding to exploit the powerful discrete Gaussian machinery. This approach avoids sample amplification and turns Wagner's algorithm into an approximate discrete Gaussian sampler for $q$-ary lattices. 
    
    For an SIS lattice with $n$ equations modulo $q$, this algorithm runs in subexponential time $\exp(O(n/\log \log n))$ to reach a Gaussian width parameter $s = q/\mathrm{polylog}(n)$ only requiring $m = n + \omega(n/\log \log n)$ many SIS variables. This directly provides a provable algorithm for solving the Short Integer Solution problem in the infinity norm ($\mathrm{SIS}^\infty$) for norm bounds $\beta = q/\mathrm{polylog}(n)$. 
    This variant of SIS underlies the security of the NIST post-quantum cryptography standard {\textsc{Dilithium}}. Despite its subexponential complexity, Wagner's algorithm does not appear to threaten \textsc{Dilithium}'s concrete security.
\end{abstract} 

\paragraph{Keywords.} Wagner's algorithm, Short Integer Solution problem (SIS), Discrete Gaussian sampling, Lattice-based cryptography, Cryptanalysis
 
%\vspace{-1mm}
\newpage

\section{Introduction}

The Short Integer Solution problem (SIS) is a fundamental problem in lattice-based cryptography. It was introduced by Ajtai~\cite{Ajt96}, along with an average-case to worst-case reduction from SIS (in the average case) to the problem of finding a short basis in an arbitrary lattice (in the worst case). 
Since then, a plethora of cryptographic schemes have been based on the presumed average-case hardness of SIS. 
The SIS problem asks to find an integer vector of norm at most $\beta$ that satisfies a set of $n$ equations in $m$ variables modulo $q$. %, where shortness is typically considered with respect to the Euclidean norm, but can also be formally defined for other norms. 
The average-case to worst-case reductions \cite{Ajt96,MR2004} consider SIS in the Euclidean norm, with norm bound $\beta$ significantly smaller than $q$. 
However, when it comes to practical cryptographic schemes, designers are often tempted to consider SIS instances outside of the asymptotic coverage of the reduction, or even to consider variants of the problem. This is the case for the new NIST standard \textsc{Dilithium}~\cite{DKLL18}, which considers SIS in the $\ell_\infty$-norm with norm bound $\beta$ rather close to the modulus $q$. 

The dual of the latter SIS variant is commonly known as the Learning with Errors problem (LWE)~\cite{Regev2005} with narrow error distribution, say ternary errors. This version of LWE was proven to be solvable in slightly subexponential time by Kirchner and Fouque~\cite{KF15} using a variant of the Blum-Kalai-Wasserman algorithm (BKW)~\cite{BKW03}, at least when the number of samples $m$ is large enough. 
More precisely, it was claimed in~\cite{KF15} that $m$ linear in $n$ suffices, but a subsequent work of Herold, Kirshanova, and May~\cite{HKM18} found an issue in the proof. They resolved the issue only for $m \geq C n\log n$ (for some constant $C>0$), leaving open whether fewer samples would suffice. 
Whether this limitation on the number of samples is fundamental or an artifact of the proof remains unclear. 

Our work was triggered by the folklore interpretation of BKW as a combination of Wagner's algorithm~\cite{Wag02}, run over the dual lattice, and the Aharonov-Regev distinguisher~\cite{AR05}. This raises a natural question: using the tricks of Kirchner and Fouque but only in the Wagner part of the algorithm, can one (provably) solve interesting instances of SIS in subexponential time?  

\subsection{Contribution} 

We answer the above question positively, namely we prove that a variant of Wagner's algorithm solves SIS with $n$ equations modulo $q = \poly(n)$, in subexponential time $\exp(O(n/\log \log n))$ for an $\ell_\infty$-norm bound $\beta = q/\polylog(n)$ and only requiring $m = n + \omega(n/\log \log n)$ many SIS variables. We also achieve subexponential complexity for ISIS and SIS$^\times$ for an $\ell_2$-norm bound $\beta = \sqrt n \cdot q/\polylog(n)$. All prior algorithms (including heuristic ones) for these problems had exponential asymptotic complexity $\exp(\Omega(n))$.  

Beyond these applications, we also provide a significant revision of the ideas and techniques at hand. First, we propose a more abstract interpretation of Wagner for lattices using a zigzag through a chain of projected lattices and superlattices, depicted in Figure~\ref{fig: lift}. This emphasizes that the very principle of the algorithm is not tied to SIS, and its subexponential complexity really stems from the ease of constructing the appropriate chain of lattices in the case of SIS lattices.

Furthermore, we circumvent the sample-amplification technique~\cite{Lyu05}, and instead follow in the footsteps of~\cite{ADRS2015,ADS15,AS18,ALS21,ACKS21} by relying on discrete-Gaussian techniques to obtain provable results. Specifically, we use discrete Gaussian sampling and rounding to maintain precise control of the distribution of vectors throughout the chain of lattices. 

Finally, we consider whether this approach threatens the concrete security claim of the NIST standard~\textsc{Dilithium}. It turns out that, despite its subexponential complexity, and removing all the overhead introduced for provability, this algorithm is far less efficient against the concrete parameters of~\textsc{Dilithium} than standard lattice reduction attacks.

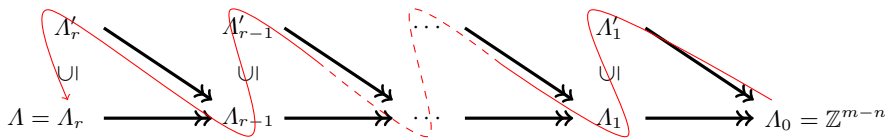
\begin{figure}[ht!]
    \centering
    \input{tikz/lift}
    \caption{Wagner's algorithm interpreted over a chain of projected lattices $\Lambda_i$ and auxiliary superlattices $\Lambda'_i$. Thick black double-arrows denote surjective orthogonal projections between lattices. The red curvy arrow denotes the path taken by the algorithm. It starts with (short) vectors in  $\Z^{m-n}$, and follows the zigzag path until it generates (short) vectors in the lattice $\Lambda$, which, in our application, correspond to SIS solutions.}
    \label{fig: lift}
\end{figure}

\vspace{-0.2cm}

\subsection{Technical Overview} \label{sec: technical introduction}

% Notation defined in the text: 
% - $\Id_n$ is the $n\times n$ identity matrix
% - $(\xv;\yv)$ denotes the column vector obtained by concatenating two column vectors $\xv, \yv$ 
% - Truncation $\xv_{i:j}:=(x_i, \dots, x_j)$ for $i<j$

The Short Integer Solution problem in the infinity norm is defined as follows. 

\begin{problem}[$\SISinfty_{n,m,q,\beta}$] \label{def: SISinfty}
Let $n, m, q \in \N$ with $m\geq n$, and let $\beta >0$. Given a uniformly random matrix $\Am \in \Z_q^{n \times m}$, the problem $\SISinfty_{n,m,q,\beta}$ asks to find a nonzero vector $\xv \in \Z^m$ satisfying $\Am \xv = \zerovec \bmod q$ and $\|\xv\|_\infty \leq \beta$. % NB: WLOG beta is integer 
\end{problem}

This problem can be phrased as a short vector problem, as it is equivalent to finding a nonzero vector of norm at most $\beta$ in the $q$-ary lattice $\Lambda_q^\bot(\Am) := \{\xv\in \Z^{m} : \Am\xv = \zerovec \bmod q\}$. 
It is trivial when $\beta > q$, as $(q, 0, \dots , 0)$ is a solution, and becomes vacuously hard (i.e., typically no solution exists) when $\beta$ is too small. % Namely if $\beta < \lambda_1(\Lambda^\bot_q(\Am))$ (but lambda_1 undefined)
% Writing `typically', because A could be the zero matrix 

Consider an instance of the $\SISinfty_{n,m,q,\beta}$ problem for some given $\Am \in \Z_q^{n \times m}$. When $\Am$ is of full rank\footnote{This happens with overwhelming probability $1 - (1/q)^{o(1)}$ when $m - n = \omega(1)$: a uniformly random matrix in $\Z_q^{n\times m}$ (with $m \geq n$) is of full rank with probability at least $1-1/q^{m-n}$.} we can, without loss of generality, assume that $\Am$ is written in systematic form, i.e., $\Am = [\Am' \mid \Id_n]$ with $\Am' \in \Z_q^{n \times (m-n)}$.
With this writing of $\Am$, we see that the problem of finding $\xv \in \Z_q^m$ satisfying $\Am \xv = \zerovec \bmod q$ and $\|\xv\|_\infty \leq \beta$ is equivalent to finding $\zv \in \Z_q^{m-n}$ satisfying $\|\xv(\zv)\|_\infty \leq \beta$ where $\xv(\zv) := (\zv ; - \Am' \zv)$.\footnote{We use the notation $(\av; \bv)$ for the vertical concatenation of vectors $\av$ and $\bv$.}

\subsubsection{Wagner-Style Algorithms for SIS.} 
Wagner's generalized birthday algorithm~\cite{Wag02} addresses the problem of finding elements from a list $L_0 \subseteq \Z_q^n$ that sum up to zero.\footnote{Although Wagner originally considered the case $\Z_2^n$, the algorithm for $\Z_q^n$ follows the same principle.} The same technique was independently used in~\cite{BKW03} to solve a dual problem. 

Given an $\SIS$ instance $\Am = [\Am' \mid \Id_n] \in \Z_q^{n \times m}$, divide the rows of the parity-check matrix $\Am$ into $r$ blocks of equal size $n/r$, as illustrated in \Cref{fig: matrix}. The algorithm then iteratively solves smaller $\SIS$ instances corresponding to the parity-check matrices $\Am_i = [\Am_i' \mid \Id_{in/r}]$, for $i=1$ up to $r$, where $\Am_i'$ is defined by the first $i$ blocks of rows of $\Am'$. % NB: `smaller' is confusing: it's smaller than full instance, but they get larger in each iteration 
Specifically, the algorithm begins by filling an initial list $L_0$ with short vectors. 
Then, in iteration $i$, it computes, for each vector $\xv \in L_{i-1} \subseteq \Z_q^{m-n+ (i-1)n/r}$, the unique (modulo $q$) vector $\yv \in \Z^{n/r}$ such that $\Am_{i} (\xv ; \yv) = \zerovec \bmod q$. 
This vector $\yv$ is likely to be of high norm since it is uniformly distributed modulo~$q$. 
The algorithm then searches for pairs of such vectors $\xv'_1 = (\xv_1;\yv_1)$, $\xv'_2 = (\xv_2;\yv_2)$ that satisfy $\yv_1=\yv_2 \bmod q$, to then add $\xv'_1-\xv'_2 = (\xv_1-\xv_2; \zerovec)$ to the list $L_i$. This ensures that the difference $\xv'_1-\xv'_2$ remains short, as $\xv_1$ and $\xv_2$ are short and the rest has norm zero. 
The final list $L_r$ contains vectors $\xv \in \Z_q^{m}$ that satisfy $\Am \xv = \zerovec \bmod q$, providing potential nonzero and short solutions to the original $\SIS$ problem. % NB: The vectors have bounded norm $\|\xv\|_\infty \leq 2^{r-1}$ 
%The complexity of this algorithm is $q^{n/r}$, and a clever choice of $r$ and the initial vectors allows to control the norm of the output vectors. 

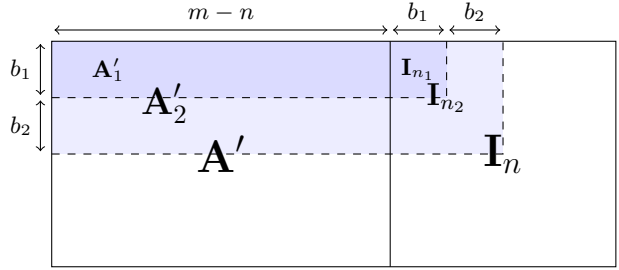
\begin{figure}[ht!]
    \centering
    \input{tikz/matrix}
    \caption{Illustration of the parity-check matrices $\Am_i := [\Am'_i \mid \Id_{\nsum_i}]$, where $\Am'_i$ is the matrix defined by the first $\nsum_i := \sum_{j=1}^i \nnew_j$ rows of $\Am'$. (Without the generalized lazy-modulus switching technique, the $\nnew_j$ are set equal to $n/r$.) Each iteration $i$ of the algorithm solves an $\SIS$ instance given by parity-check matrix $\Am_i$, until ultimately the whole matrix $\Am = \Am_r$ is covered.}
    \label{fig: matrix}
\end{figure}

\paragraph{Lazy-Modulus Switching.} 
The aforementioned approach constructs the lists $L_i$ by combining vectors that lie in the same `bucket', where the bucket of $(\xv;\yv)$ for $\xv \in L_{i-1}$ is labeled by the value of $\yv$ modulo~$q$. The vectors in $L_i$ are then guaranteed to be of the form $(\xv' ; \zerovec)$, but this is not necessary for the algorithm to succeed: a form like $(\xv' ; \yv')$ with both $\xv'$ and $\yv'$ being short suffices. 
Based on this observation, the authors of~\cite{AFFP14} consider using a smaller modulus $p<q$, and combine $(\xv_1;\yv_1)$ and $(\xv_2;\yv_2)$ if $\lfloor \frac p q \yv_1 \rceil= \lfloor \frac p q \yv_2 \rceil \bmod p$. 
This modified condition for combining vectors may result in a nonzero `rounding error' in the $\yv'$-part, which is traded for a reduced number of buckets: $p^{n/r}$ instead of $q^{n/r}$.  

Two concurrent works~\cite{KF15,GJMS17} generalize this `lazy-modulus switching' technique by considering different moduli $p_1, \dots, p_r$ (not necessarily prime). 
In this approach, the matrix $\Am$ is divided into $r$ blocks of respective size $\nnew_1, \dots, \nnew_r$. 
The rounding errors induced by iteration $i$ are at most $q/p_i$, and may double in each subsequent iteration when the vectors are combined. Hence, it makes sense to use decreasing moduli $p_i$ to balance the accumulation of rounding error. 
Each step requires (more than) $p_i^{\nnew_i}$ vectors to find collisions in $\Z_{p_i}^{\nnew_i}$, leading to increasing block sizes $\nnew_i$. This increasing choice of $\nnew_i$'s is central to the subexponential complexity claim of~\cite{KF15}. 

\subsubsection{A Naive Analysis.}

Let us consider \cref{alg: worst-case Wagner for SIS} and attempt to analyze it in order to highlight the core of the issue. It is a rephrased version of the algorithm of~\cite{KF15} without the LWE dual-distinguishing step.

\begin{algorithm}
\caption{Wagner-Style Algorithm for $\SIS$} 
\label{alg: worst-case Wagner for SIS}
\DontPrintSemicolon
\SetKwInOut{Input}{Input}\SetKwInOut{Output}{Output}

\Input{Integers $n,m,q$; \\
    Full-rank matrix $\Am = [\Am' \mid \Id_n] \in \Z_q^{n \times m}$; \\ 
    Integer parameters $N, r, (p_i)_{i=1}^r, (\nnew_i)_{i=1}^r$ with $\sum_{i=1}^r \nnew_i = n$} 
\Output{List of vectors $\xv \in \Z_q^m$ such that $\Am \xv = \zerovec \bmod{q}$ and $\|\xv\|_\infty \leq 2^r$}
\vspace{2mm}

Initialize a list $L_0$ with $3^r N$ independent, uniform samples from $\{-1,0,1\}^{m-n}$ \\
\For{$i = 1, \ldots, r$}{
    $L_i := \emptyset$ \\
    Initialize empty buckets $B(\cv)$ for each $\cv \in \Z_{q}^{\nnew_i}$ \\
    \For(\tcp*[f]{Bucketing}){$\xv \in L_{i-1}$}{
    Compute the unique $\yv \in \Z_{q}^{\nnew_i}$ satisfying $\Am_i(\xv;\yv) = \zerovec \bmod q$ \\ 
    Compute $\cv = \left\lfloor \frac{p_i}{q} \yv \right\rceil\bmod p_i$ \\
    Append $\xv' := (\xv ; \yv)$ to $B(\cv)$ \\
    }
    \For(\tcp*[f]{Combining}){$\cv \in \Z_{p_i}^{\nnew_i}$}{
        \For{each two vectors $\xv'_1, \xv'_2$ in $B(\cv)$}{
            Append $\xv'_1 - \xv'_2$ to the list $L_i$ \\ 
            Remove $\xv'_1$ and $\xv'_2$ from $B(\cv)$ \\
        }
    }
}
\Return{$L_r$}
\end{algorithm}

\paragraph{List Construction.} Since each vector in $L_i$ is either paired or left alone in a bucket, it follows that $|L_i| \leq 2|L_{i+1}| + p_i^{\nnew_i}$. Thus, by initializing a list $L_0$ of size $3^r N$, we obtain at least $N$ vectors in $L_r$, assuming $N \geq \max_{i} p_i^{\nnew_i}$.

\paragraph{Bucketing and Combining.}
At the beginning of each iteration $i \in \{1,\dots,r\}$, the algorithm initializes some empty lists $B(\cv)$, the `buckets', each labeled by a vector (coset) $\cv \in \Z_{p_i}^{\nnew_i}$. 
Every iteration works in two phases. First, the `for'-loop over the $\xv \in L_{i-1}$ performs the \textit{bucketing} phase: the vectors are added to a bucket according to their corresponding $\yv$-value. 
In the \textit{combining} phase, the `for'-loop over each bucket representative $\cv \in \Z_{p_i}^{\nnew_i}$ takes differences of vectors $\xv'_1, \xv'_2$ belonging to the same bucket. % NB: Notation $\xv'=(\xv,\yv)$ may be confusing (confusion of $\xv$ and $\xv'$), but makes a direct parallel with lattices $\Lambda$ and $\Lambda'$ 
Since $\Am\xv'_1 = \zerovec \bmod q$ and $\Am\xv'_2 = \zerovec \bmod q$, the difference vector $\xv'_1 - \xv'_2$ also satisfies $\Am(\xv'_1-\xv'_2) = \zerovec \bmod q$ by linearity. 
If the vectors $\xv'_1 = (\xv_1;\yv_1)$, $\xv'_2 = (\xv_2;\yv_2)$ belong to the same bucket $B(\cv)$ for a given $\cv$, then they satisfy $\cv = \lfloor \frac{p_i}{q} \yv_1 \rceil = \lfloor \frac{p_i}{q} \yv_2 \rceil \bmod p_i $, where the operation $\lfloor \cdot \rceil$ denotes rounding each coordinate to its nearest integer modulo $p_i$. 
Consequently, $\yv_1$ and $\yv_2$ are `close' to each other, ensuring their difference is short: we have that $\|\yv_1 - \yv_2\|_\infty \leq q/p_i$, and thus $\|\xv'_1-\xv'_2\|_\infty \leq \max\{2\|\xv_1\|_\infty, 2\|\xv_2\|_\infty, q/p_i\}$. 
By induction over the  $r$ iterations, it follows that the algorithm outputs vectors $\xv \in L_r$ satisfying
\begin{align}\label{eq: bound on norm tech intro}
    \|\xv\|_\infty \leq \max_{0 \leq i \leq r} 2^{r-i} \tfrac{q}{p_i}, 
\end{align}
where we define $p_0 = q$.

\subsubsection{Time Complexity of \Cref{alg: worst-case Wagner for SIS}.}

\Cref{eq: bound on norm tech intro} shows that the choice of the number of iterations $r$ and the moduli $p_i$ influences the maximal norm of the vectors. These parameters also affect the algorithm's time complexity, along with the other parameters that need to be chosen. Although the algorithm parameters $p_i, \nnew_i, r$ and $N$ should be integers to make sense, we consider them as real numbers for simplicity in this introductory section. 

\paragraph{Parameter Selection.}
Let the target norm for the $\SIS$ problem be $\beta = \frac{q}{f}$ for some factor $f > 1$. 
Note that the problem is easiest when $f=1$ (norm $q$) and harder when $f$ increases (i.e., the norm $q/f$ becomes shorter). 
The bound on the output norms given in \Cref{eq: bound on norm tech intro} is minimized when both terms are balanced: $2^r = 2^{r-i}q/p_i$. Therefore we set $p_i: = q/2^i$. 
One can generate $p_i^{\nnew_i}$ distinct vectors modulo $p_i$ with $\nnew_i$ coordinates, so to keep the number of samples comparable at each step $i$, we set $N = p_i^{\nnew_i}$. 
This implies $\nnew_i = \frac{\log N}{\log p_i} = \frac{\log_2 N}{\log_2 q - i}$. 
By the choice of the $p_i$, the final vectors have a norm of at most $2^r$. To ensure that it is at most $\beta = q/f$, we must choose $r$ such that $2^r \leq \frac{q}{f}$. We set $r := \log_2(q/f) - 1$.  
We must also ensure that $n = \nsum_r = \sum_{i=1}^r \nnew_i$. 
Since the $\nnew_i$ increase with $i$, we can bound their sum by $n = \sum_{i=1}^r \nnew_i \leq \int_1^{r+1} \frac{\log_2 N}{\log_2 (q) - x} \text{d}x = (\log_2 N) \cdot \ln \left( \frac{\log_2 q-1}{\log_2 q-(r+1)} \right)$.\footnote{For $f$ an increasing function, $\int_{0}^{r} f(x) \,\text{d}x \leq \sum_{i=1}^r f(i) \leq \int_{1}^{r+1} f(x)\,\text{d}x$. Also, for $A,B,a,b > 0$, $\int_a^b \frac{A}{B-x} \text{d}x = A\ln \left(\frac{B-a}{B-b}\right)$.\label{footnote: integral facts}} Thus, % By definition of r
\begin{equation}\label{eq:N Riemann worst case}
    n \leq \log_2(N) \cdot \ln \left( \frac{\log_2 q}{\log_2 f} \right). 
\end{equation}
Rewriting, we conclude that taking $\log_2(N)= \frac{n}{\ln \ln(q) - \ln \ln(f)}$ suffices. Up to rounding of the parameters (as they need to be integers), we would get the following statement. 

\textbf{Tentative Theorem.} \textit{Let $n,m,q \in \N$ and $f>1$. There exists an algorithm that, given $\Am \in \Z_q^{n \times m}$, returns a (possibly zero) vector $\xv \in \Z_q^m$ such that $\Am\xv = \zerovec\bmod q$ and $\|\xv\|_\infty \leq \frac{q}{f}$ in time} $$T = 2^{\tfrac{n}{\ln \ln(q) - \ln \ln(f)}} \cdot \poly(n, \log q).$$ 

\paragraph{Important Remark.} There is a catch!
The output list of the algorithm contains a $\SIS$ solution only if at least one of the output vectors is \textit{nonzero}. So it remains to be proven that the vectors do not all cancel out to zero. 
At the first iteration, the set $\{-1,0,1\}^{m-n}$ is significantly larger than the number of buckets, hence differences of vectors from a same bucket are unlikely to be zero vectors. However, from the second iteration onward, specifying the vector distributions is far from straightforward. 

This issue does not arise in the original works of~\cite{Wag02,BKW03} where $m$ was assumed to be as large as the initial list size $|L_0|$, and the initial list is simply filled with standard unit vectors (i.e., vectors with one coordinate equal to 1 and zeros elsewhere). These vectors are linearly independent and hence cannot cancel each other out. In the case of solving ternary LWE with~\cite{BKW03} rather than SIS with~\cite{Wag02}, this situation can be emulated using a sample-amplification technique from~\cite{Lyu05}, however, this requires at least $m = \Theta(n \log n)$ samples to start with according to~\cite{HKM18}; the argument of~\cite{KF15} that this should work for $m = \Theta(n)$ has been shown to be flawed in~\cite{HKM18}.  

Yet, at least heuristically, it is not clear why this approach should fail when $m = \Theta(n)$. There is enough entropy to avoid collision at the first step, and entropy should intuitively increase at a later stage: after all, the vectors are getting larger. But to adequately formalize this intuition, we shift our perspective on the algorithm.

\subsubsection{From Parity-Check Perspective to Lattices.}  \label{sec: geometric interpretaion}

As mentioned before, the $\SIS$ problem is equivalent to a short vector problem which asks to find a nonzero vector of norm at most $\beta$ in the lattice $\Lambda:= \{\xv\in \Z^{m}: \Am\xv = \zerovec \bmod q\}$.  
We consider the sequence of lattices $\Z^{m-n} = \Lambda_0, \dots, \Lambda_r = \Lambda$ with $\Lambda_i := \{\xv\in \Z^{m-n+\nsum_i} : \Am_i\xv = \zerovec \bmod q\}$. 
Each parity-check matrix $\Am_i$ adds a block of $\nnew_i$ new coordinates, making $\Lambda_{i-1}$ a projection of $\Lambda_i$. 
It is easy to sample bounded vectors in $\Lambda_0 := \Z^{m-n}$, by sampling for example in $\{-1,0,1\}^{m-n}$, as \Cref{alg: worst-case Wagner for SIS} did. 
The goal is to use these initial samples to go back through the projections towards $\Lambda$ (recall \Cref{fig: lift}), while keeping the norms bounded. 

While one can lift back a vector from $\Lambda_{i-1}$ to $\Lambda_{i}$, such a lifted vector would not be short because the lattice $\Lambda_{i}$ we are lifting over is too sparse. Instead, what happens in \Cref{alg: worst-case Wagner for SIS}, is that we implicitly lift to a denser lattice $\Lambda'_{i} \supseteq \Lambda_{i}$, and then take the difference of two vectors in the same coset of the quotient $\Lambda'_{i}/\Lambda_{i}$ to fall again in $\Lambda_{i}$. More explicitly, setting $\cv = \lfloor \frac{p_i}{q} \yv \rceil$ as in  \Cref{alg: worst-case Wagner for SIS}, note that the vector $(\xv, \yv - \frac q {p_i} \cv)$ lives in the lattice $\Lambda'_i = \Lambda_i + (\{0\}^{m - n + \nsum_{i-1}} \times \frac q {p_i} \Z^{\nnew_i} )$ which satisfies $|\Lambda'_i/\Lambda_i| = p_i^{\nnew_i}$. Furthermore, the coset of that vector in the quotient $\Lambda'_i/\Lambda_i$ is exactly determined by ${\cv} \bmod p$.   

This provides a clean and pleasant interpretation of the algorithm as walking in a commutative diagram shown in \Cref{fig: lift}. It further provides the framework to deploy the full machinery of discrete Gaussian distributions over lattices: the initial samples can easily be made Gaussian over $\Z^{m-n}$, the lifting step as well using the randomized Babai algorithm~\cite{GPV2008,EWY23}, and the combination step preserves Gaussians using convolution lemmas~\cite{Pei2010,MP2013}, as already used in~\cite{ADRS2015,ADS15,ACKS21,AS18,ALS21} to solve short vector problems. There are further technicalities to control the independence between the samples, which we handle using the (conditional) similarity notion introduced in~\cite{ALS21}. This permits to control all the distributions throughout the algorithm, leading to provable conclusions. 
In particular, a careful choice of the algorithm's parameters guarantees, with high probability, that the final distribution is not concentrated at zero~\cite[Lemma~2.11]{PR2006}.

\subsection{Organization of the Paper} 
\Cref{sec:preliminaries} provides the necessary background for this paper. In \Cref{sec: Wagner Gaussian sampler (full section)}, we present our Gaussian sampler and analyze its asymptotic time complexity. In \Cref{sec:Asymptotic Application to Cryptographic Problems}, we use the Gaussian sampler to asymptotically solve several variants of SIS, carefully avoiding the `canceling out to zero' issue. Finally, \Cref{sec:Concrete and Heuristic Application to Dilithium} discusses the impact of the attack on the concrete security of \textsc{Dilithium}.

\subsection{Acknowledgements}
The authors are grateful to Ronald de Wolf and the anonymous reviewers for their useful comments on the manuscript. 
LD and JL were supported by the ERC Starting Grant 947821 (ARTICULATE).
LE was supported by the Dutch National Growth Fund (NGF), as part of the Quantum Delta NL program.

\section{Preliminaries} \label{sec:preliminaries}

\paragraph{Notation.}
% Bold lower and upper cases respectively denote vectors and matrices. 
We write $\|\cdot\|_2$ and $\|\cdot\|_\infty$ for the Euclidean and infinity norm of a vector, respectively, and define $\mathcal{B}_n^\infty := \{\xv \in \R^n \colon \|\xv\|_\infty \leq 1\}$.  
For a positive integer $N$, we define $[N] := \{1,\ldots, N\}$. We use the notation $X = e^{\pm \delta}$ as shorthand for $X \in [e^{-\delta},e^{\delta}]$. 
For $x = (x_1,\ldots, x_N)$ and $i\in [N]$, we define $x_{-i} := (x_1, \ldots, x_{i-1}, x_i, \ldots, x_N)$. We define $x_{-\{i,j\}}$ analogously.
For events $E_0,E_1$, we use the convention that $\Pr[E_0 \mid E_1] = 0$ if $\Pr[E_1] = 0$.

\paragraph{Asymptotic Notation.} 
Let $f$ and $g$ be functions that map positive integers to positive real numbers. 
We write $f(n) = O(g(n))$ if there exist constants $c, n_0 > 0$ such that $f(n) \leq c \cdot g(n)$ for every integer $n \geq n_0$. 
We write $f(n) = \Omega(g(n))$ if there exist constants $c, n_0 > 0$ such that $f(n) \geq c \cdot g(n)$ for every integer $n \geq n_0$. 
We write $f(n) = \Theta(g(n))$ if both $f(n) = O(g(n))$ and $f(n) = \Omega(g(n))$. 
We define $\tilde{O}(f(n)) := O\left(f(n) \cdot \mathrm{polylog}(f(n))\right)$, where $\mathrm{polylog}(f(n)) := \log(f(n))^{O(1)}$. 
Similarly, $\tilde{\Omega}(f(n)) := \Omega\left( f(n) / \mathrm{polylog}(f(n)) \right)$. 
We write $f(n) = o(g(n))$ if, for all constants $c > 0$, there exists $n_0 > 0$ such that $f(n) < c \cdot g(n)$ for every integer $n \geq n_0$. 
We write $f(n) = \omega(g(n))$ if, for all constants $c>0$, there exists $n_0 > 0$ such that $f(n) > c \cdot g(n)$ for every integer $n \geq n_0$.

\subsection{Similarity of Distributions}

Inspired by~\cite{ALS21}, we consider the notions `similarity' and `conditional similarity' to measure the pointwise distance between two distributions. These concepts are slightly stronger than statistical distance (see \Cref{rem: similarity implies closeness}), and are particularly useful for handling small independencies arising from `bucket and combine' type of algorithms (like each iteration of Wagner-style algorithms).

\begin{definition}[Similar]\label{defn: similarity}
    Let $D$ be a probability distribution over a set $\mathcal{X}$, and let $\delta \geq 0$ be a real. 
    A random variable $X \in \mathcal{X}$ is $\delta$-similar to $D$ if, for all $x \in \mathcal X$, it holds that \begin{align*}
        \Pr_{X}[X = x] = e^{\pm \delta} \cdot \Pr_{Y \sim D}[Y = x].
    \end{align*}
\end{definition}

\begin{remark}[Similarity implies Closeness in Statistical Distance]\label{rem: similarity implies closeness}
    The notion of $\delta$-similarity is a stronger notion than being within statistical distance $\delta$, where we recall that the statistical distance between two discrete random variables $X, Y$ is defined as $\frac{1}{2} \sum_{x \in \mathcal X} \left|\Pr[X=x] - \Pr[Y=x]\right|$. Indeed, for all $\delta \in [0,1]$, being $\delta$-similar implies being within statistical distance $\delta$, since for any such $\delta$ we have $[e^{-\delta}, e^{\delta}]  \subseteq [1 - 2\delta, 1 + 2\delta]$.
\end{remark}

The next definition is closely related to~\cite[Definition~4.1]{ALS21}, but we remark that we use different terminology and that we made the definition more general. % Arbitrary distributions instead of discrete Gaussians 

\begin{definition}[Conditionally Similar]
    Let $D$ be a probability distribution over a set $\mathcal{X}$, and let $\delta \geq 0$ be a real. 
    Discrete random variables $X_1, \ldots, X_N \in \mathcal{X}$ are \emph{conditionally $\delta$-similar to independent samples from $D$} if, for all $i\in [N]$ and $x \in \mathcal{X}^N$, it holds that 
    \begin{align*}
        \Pr_{X_1,\ldots, X_N}[X_i = x_i \mid X_{-i} = x_{-i}]  = e^{\pm \delta} \cdot \Pr_{Y \sim D}[Y = x_i].
    \end{align*}
\end{definition}

In particular, being conditionally $0$-similar (i.e., $\delta = 0$) is equivalent to being independently distributed according to $D$.

\begin{remark}[Conditional Similarity implies Marginal Similarity]\label{rem: conditional similarity implies marginal similarity}
    For $N~=~1$, conditional similarity coincides with \Cref{defn: similarity}. More generally, for all $N\geq 1$ and $\delta \geq 0$, conditional $\delta$-similarity implies marginal $\delta$-similarity.  
    Indeed, if $X_1,\ldots,X_N \in \mathcal X$ are conditionally $\delta$-similar to independent samples from a distribution $D$ on $\mathcal X$, then we have, for all $i \in [N]$ and $x \in \mathcal{X}$, that $\Pr[X_i = x] = \sum_{y \in \mathcal{X}^{N-1}}\Pr[X_i = x | X_{-i} = y] \Pr[X_{-i} = y] = e^{\pm \delta} \cdot \Pr_{Y \sim D}[Y=x]$.  
\end{remark}

A useful property of similarity and conditional similarity is that these notions are closed under convex combinations (as already observed in~\cite{ALS21}). 
   
% The lemma is also true if we allow $\mathcal{E}$ to depend on $i$.
\begin{lemma}\label{lem: similarity and conditional similarity closed under convex combinations}
    Let $D$ be a probability distribution over a set $\mathcal{X}$, and let $\delta \geq 0$ be a real. 
    Let $\mathcal{E}$ be a finite or countably infinite set of mutually exclusive and collectively exhaustive events. 
    If discrete random variables $X_1, \ldots, X_N \in \mathcal{X}$ satisfy
    \begin{align*}
        \Pr_{X_1,\ldots, X_N}[X_i = x_i \mid X_{-i} = x_{-i} \text{ and } E]  = e^{\pm \delta} \cdot \Pr_{Y \sim D}[Y = x_i]
    \end{align*}
    for all $i \in [N]$, $x \in \mathcal{X}^N$, and $E \in \mathcal{E}$, then $X_1, \ldots, X_N$ are conditionally $\delta$-similar to independent samples from $D$. 
    % Similarity (N=1): 
    % In particular, if a random variable $X \in \mathcal{X}$ satisfies \begin{align*}
    %     \Pr_{X}[X = x \mid E] = e^{\pm \delta} \cdot \Pr_{Y \sim D}[Y = x],
    % \end{align*} % That is, the distribution of X conditional on $E$ is $\delta$-similar to D for all $x \in \mathcal{X}$ and all $E \in \mathcal{E}$, then $X$ is $\delta$-similar to $D$. 
\end{lemma}

\begin{proof}
Suppose that the premise is true for random variables $X_1,\ldots, X_N$. Then, for all $i \in [N]$ and $x \in \mathcal{X}^N$, we have {\allowdisplaybreaks \begin{align*}
    &\, \Pr_{X_1,\ldots, X_N}[X_i = x_i \mid X_{-i} = x_{-i}] \\
    =&\, \sum_{E \in \mathcal{E}} \Pr_{X_1,\ldots, X_N}[X_i = x_i \text{ and } E \mid X_{-i} = x_{-i}] \\
    =&\, \sum_{E \in \mathcal{E}} \Pr_{X_1,\ldots, X_N}[X_i = x_i \mid E \text{ and } X_{-i} = x_{-i} ] \cdot \Pr_{X_1,\ldots, X_N}[E \mid X_{-i} = x_{-i}] \\
    =&\, e^{\pm \delta} \cdot \Pr_{Y \sim D}[Y = x_i]
\end{align*}
since $\sum_{E \in \mathcal{E}} \Pr[E \mid X_{-i} = x_{-i}] = \frac{\sum_{E \in \mathcal{E}} \Pr[E \text{ and } X_{-i} = x_{-i}]}{\Pr[X_{-i} = x_{-i}]}  = 1$.}
This proves the lemma. \qed
\end{proof}

The following property can be viewed as the data-processing inequality for conditional similarity.  
\begin{lemma}\label{lem: easy consequence of conditional similarity}
    Let $D$ be a probability distribution over a set $\mathcal{X}$, and let $\delta \geq 0$ be a real. If discrete random variables $X_1, \ldots, X_N \in \mathcal{X}$ are conditionally $\delta$-similar to independent samples from $D$, 
    then \begin{align*}
        \Pr_{X_1,\ldots, X_N}[f(X_i) = 1] = e^{\pm \delta} \cdot \Pr_{Y \sim D}[f(Y) = 1]
    \end{align*}
    for all $i \in [N]$ and all functions $f \colon \mathcal{X} \to \{0,1\}$. 
\end{lemma}

\begin{proof}
    Suppose that $X_1,\ldots,X_N$ are conditionally $\delta$-similar to independent samples from $D$. Let $i \in [N]$ and let $f \colon \mathcal{X} \to \{0,1\}$ be an arbitrary function.
    Define $f^{-1}(1) := \{x \in \mathcal X \colon f(x) = 1\}$. 
    Then \begin{align*}
        \Pr_{X_1,\ldots, X_N}[f(X_i) = 1] &= \sum_{x \in \mathcal X} \Pr_{X_1,\ldots, X_N}[f(X_i) = 1 \text{ and } X_i = x] \\
        %&= \sum_{x \in \mathcal X} \Pr_{X_1,\ldots, X_N}[X_i = x] \cdot \Pr_{X_1,\ldots, X_N}[f(X_i) = 1 \mid X_i = x] \\
        &= \sum_{x \in f^{-1}(1)} \Pr_{X_1,\ldots, X_N}[X_i = x] \\
        &= e^{\pm \delta} \cdot \sum_{x \in f^{-1}(1)} \Pr_{Y \sim D}[Y = x] \\
        &= e^{\pm \delta} \cdot \Pr_{Y \sim D}[f(Y) = 1]. 
    \end{align*} \qed
\end{proof}

\subsection{Lattices}

Given $k$ linearly independent vectors $\bv_1, \dots, \bv_k \in \R^n$, let $\Bm \in \R^{n \times k}$ be the matrix whose columns are the $\bv_i$. The \textit{lattice} associated to $\Bm$ is the set $\mathcal{L}(\Bm) := \Bm \Z^k =  \left\{ \sum_{i=1}^k z_i \bv_i : z_i \in \Z \right\} \subseteq \R^n$ of all integer linear combinations of these vectors. We say that $\Bm$ is a \textit{basis} for a lattice $\mathcal{L}$ if $\mathcal{L} = \mathcal{L}(\Bm)$. We say that $\mathcal{L}$ has \textit{rank} $k$ and \textit{dimension} $n$. If $n=k$, then $\mathcal{L}$ is said to be of \textit{full rank}.
We define $\Span_\R(\mathcal{L}(\Bm)) = \Span_\R(\Bm) = \{\Bm \xv \colon \xv \in \R^n\}$.
We define the \textit{dual} of $\mathcal{L}$ by $\mathcal{L}^* := \{ \yv \in \Span_\R(\mathcal{L}) \colon \forall \xv \in \mathcal{L}, \langle \xv, \yv \rangle \in \Z\}$, which is a lattice.

\paragraph{First Successive Minimum.} Given a lattice $\mathcal{L}$, we write $\lambda_1(\mathcal{L}):=\inf\{\|\xv\|_2 \colon \xv \in \mathcal{L} \setminus \{\zerovec\}\}$ for the (Euclidean) norm of a shortest lattice vector. We define $\lambda_1^\infty(\mathcal{L})$ similarly for the infinity norm.

\paragraph{Projections and Primitive Sublattices.} 
For any $S \subseteq \R^n$, we write $\pi_S$ for the projection onto $\Span_\R(S)$ and $\pi_{S}^{\bot}$ for the projection orthogonal to $\Span_\R(S)$. 
We say that a sublattice $\mathcal{S}$ of a lattice $\mathcal{L} \subseteq \R^n$ is \textit{primitive} if $\mathcal{S} = \Span_{\R}(\mathcal{S}) \cap \mathcal{L}$. It implies that there exists a sublattice $\mathcal{C} \subseteq \mathcal{L}$ such that $\mathcal{S} \oplus \mathcal{C} = \mathcal{L}$ (i.e., $\mathcal{S} + \mathcal{C} = \mathcal{L}$ and $\mathcal{S} \cap \mathcal{C} = \{\zerovec\}$). We then say that $\mathcal{C}$ is a \textit{complement} to $\mathcal{S}$.

\paragraph{Relevant $q$-ary Lattices.} 
We say that a lattice $\mathcal{L} \subseteq \R^n$ is\textit{ $q$-ary} if $q\Z^n \subseteq \mathcal{L} \subseteq \Z^n$. The two relevant $q$-ary lattices in this work are of the following form. 
For $\Am \in \Z_q^{n \times m}$, we define the full-rank $q$-ary lattices \begin{align*}
    \Lambda_q^\bot(\Am) &:= \{\xv \in \Z^m \colon \Am\xv = 0 \bmod{q}\} \subseteq \Z^{m}, \\ % `kernel lattice' 
    \Lambda_q(\Am) &:= \{\yv \in \Z^m \colon \exists \sv \in \Z_q^n, \yv = \Am^{\top}\sv \bmod{q}\} = \Am^{\top} \Z^n + q\Z^m \subseteq \Z^{m}.
\end{align*}
They are duals up to appropriate scaling: namely, $\Lambda_q(\Am) = q \cdot (\Lambda_q^\bot(\Am))^*$ and $\Lambda_q^\bot(\Am) = q \cdot (\Lambda_q(\Am))^*$. 
Furthermore, $\det(\Lambda_q^\bot(\Am)) \leq q^n$, $\det(\Lambda_q(\Am)) \geq q^{m-n}$, and $\det(\Lambda_q^\bot(\Am)) \cdot \det(\Lambda_q(\Am)) = q^m$. 

If $m \geq n$ and $\Am$ is of full rank, we can assume without loss of generality that $\Am = [\Am' \mid \Id_n]$ for some $\Am' \in \Z_q^{n \times (m-n)}$. Then a basis of $\Lambda_q^\bot(\Am)$ is given by \begin{align*}
    \begin{pmatrix}
        0 & \Id_{m-n} \\
        q \Id_{n} & -\Am'
    \end{pmatrix}.
\end{align*}

\subsection{Discrete Gaussian Distribution and Smoothness}

In the following, when the subscripts $s$ and $\cv$ are omitted, they are respectively taken to be $1$ and $\zerovec$. 

For any real $s > 0$ and $\cv \in \R^n$, we define the Gaussian function on $\R^n$ centered at $\cv$ with parameter $s$ by \begin{align*}
    \forall \xv \in \R^n,\, \rho_{s,\cv}(\xv) := \exp(-\pi \|(\xv - \cv)/s\|_2^2).
\end{align*} 

For any countable set $A$, we define $\rho_{s,\cv}(A) = \sum_{\xv \in A} \rho_{s,\cv}(\xv)$. Note that $\rho_{s,\cv}(\xv) = \rho_{s}(\xv - \cv)$, and thus $\rho_{s,\cv}(A) = \rho_{s}(A - \cv)$.

For any real $s > 0$, $\cv \in \R^n$, and full-rank lattice $\mathcal{L} \subseteq \R^n$, we define the discrete Gaussian distribution over $\mathcal{L}$ centered at $\cv$ with parameter $s$ by \begin{align*}
    \forall \xv \in \mathcal{L},\, D_{\mathcal{L},s,\cv}(\xv) := \frac{\rho_{s,\cv}(\xv)}{\rho_{s,\cv}(\mathcal{L})} = \frac{\rho_{s}(\xv - \cv)}{\rho_{s}(\mathcal{L} - \cv)}
\end{align*}
and it is 0 for $\xv\notin \mathcal{L}$. 

Similarly, for any $\tv \in \R^n$, we define  $D_{\mathcal{L} - \tv, s, \cv}(\yv) := \frac{\rho_{s, \cv}(\yv)}{\rho_{s, \cv}(\mathcal{L} - \tv)}$ for $\yv \in \mathcal{L} - \tv$. (Note that $D_{\mathcal{L},s,\cv} \equiv \cv +  D_{\mathcal{L}-\cv,s}$.)

\paragraph{Infinity Norm of Discrete Gaussian Samples.} We can tail-bound the infinity norm of a discrete Gaussian sample using the following lemma. 

\begin{lemma}[{\cite[Lemma~2.10]{Ban95}}]\label{lem: Lemma 2.10 in Ban95}
    For any full-rank lattice $\mathcal{L} \subseteq \R^n$ and real $R > 0$, 
    \begin{align*}
        \frac{\rho(\mathcal{L} \setminus R \cdot \mathcal{B}_n^\infty)}{\rho(\mathcal{L})} &< 2n \cdot e^{-\pi R^2}.
    \end{align*} 
    % Corollary 3.4 in Pei2008:  
    % In particular, if $n \geq 3$, then $\frac{\rho(\mathcal{L} \setminus \sqrt{\ln(n)} \cdot \mathcal{B}_n^\infty)}{\rho(\mathcal{L})} &< \frac{1}{4}$. 
    % (This follows from setting $R = \sqrt{\ln(8n)/\pi}$ which is $\leq \sqrt{\ln(n)}$ for $n \geq 3$.)
\end{lemma}

\paragraph{Smoothness.} 
The work of~\cite{MR2004} introduced a lattice quantity known as the smoothing parameter. 
More precisely, for any full-rank lattice $\mathcal{L} \subseteq \R^n$ and real $\eps > 0$, we define the \textit{smoothing parameter} $\eta_\eps(\mathcal{L})$ as the smallest real $s > 0$ such that $\rho_{1/s}(\mathcal{L}^*\setminus \{\zerovec\}) \leq \eps$.

% (Note that $\rho_{1/s}(\mathcal{L}^*\setminus \{\zerovec\})$ is strictly decreasing as a function of $s$. Furthermore, for a sublattice $\mathcal{L}' \subseteq \mathcal{L}$, we have $\eta_\eps(\mathcal{L}') \geq \eta_\eps(\mathcal{L})$.)

Intuitively, it gives a lower bound on $s$ such that $D_{\mathcal{L}, s}$ `behaves like' a continuous Gaussian distribution, in a specific mathematical sense. 
The following lemma justifies the name of the smoothing parameter. 
\begin{lemma}[Implicit in {\cite[Lemma~4.4]{MR2004}}]\label{lem: Lemma 4.4 in MR2004 implicit version}
    For any full-rank lattice $\mathcal{L} \subseteq \R^n$, real $\eps \in (0,1)$, real $s \geq \eta_\eps(\mathcal{L})$, and $\cv \in \R^n$, 
    \begin{align*}
        \frac{\rho_{s,\cv}(\mathcal{L})}{\rho_{s}(\mathcal{L})} \in \left[\frac{1-\eps}{1+\eps}, 1\right].
    \end{align*}
\end{lemma}

For $s$ slightly above smoothing, we can upper bound the probability of the most likely outcome of the discrete Gaussian distribution. 
\begin{lemma}[Min-entropy {\cite[Lemma~2.11]{PR2006}}] \label{lem: upper bound on maximum discrete Gaussian probability (min-entropy)} 
    For any full-rank lattice $\mathcal{L} \subseteq \R^n$,  reals $\eps > 0$ and $s \geq 2\eta_\eps(\mathcal{L})$, center $\cv \in \R^n$, and vector $\xv \in \R^n$ we have \begin{align*}
        \Pr_{X \sim D_{\mathcal{L}, s, \cv}}[X = \xv] \leq \frac{1 + \eps}{1 - \eps} \cdot 2^{-n}.
    \end{align*}
\end{lemma}

In particular, for $\xv = \zerovec$, \Cref{lem: upper bound on maximum discrete Gaussian probability (min-entropy)} gives an upper bound on the probability that a discrete Gaussian sample is zero (if the standard deviation $s$ is large enough).

\paragraph{Sampling and Combining.} 

In our algorithms, we sample from (scalings of) $\Z^n$ using the exact Gaussian sampler from~\cite{BLPRS2013}.  
% NB: We want to use the exact sampler for sampling $L_0$ (instead of the randomized GPV-sampler), because we want a guarantee on the `similarity' (pointwise distance) of the initial vectors, instead of just on the `statistical distance'. 

\begin{lemma}[Implicit in {\cite[Lemma~2.3]{BLPRS2013}}]\label{lem: exact randomized Gaussian rounding}
    There is a randomized algorithm that, given a real $s \geq \sqrt{\ln(2n + 4)/\pi}$ and $\cv \in \R^n$, returns a sample from $D_{\Z^n, s, \cv}$ in expected time $\mathrm{poly}(n, \log s, \log \|\cv\|_\infty)$. 
\end{lemma} 
% NB: The algorithm in~\cite[Lemma~2.3]{BLPRS2013} would return a sample $X$ from  $D_{\Z^n - \cv, s}$, but then $\cv + X$ is sampled according to $D_{\Z^n, s, \cv}$. 

%It follows that there is a randomized polynomial-time algorithm to sample from $D_{\alpha \Z^n, s, \cv}$ as long as $s \geq \alpha \sqrt{\ln(2n + 4)/\pi}$. Namely, use the above algorithm to sample from $D_{\Z^n, s/\alpha,\cv/\alpha}$ and multiply the result by $\alpha$. 

Our proofs use a variant of the convolution lemma~\cite[Theorem~3.3]{MP2013} (see also~\cite[Theorem~3.1]{Pei2010}) that bounds how similar the difference of two discrete Gaussians is to a discrete Gaussian. It is a slightly tighter result than~\cite[Lemma~2.14]{ALS21}. 

\begin{lemma}[Explicit Variant of Convolution Lemma]\label{lem: convolution lemma with specified distance}
    Let $\mathcal{L} \subseteq \R^n$ be a full-rank lattice and let $s \geq \sqrt{2} \eta_\eps(\mathcal{L})$ for some real $\eps > 0$. For $i = 1,2$, let $\mathcal{L} + \cv_i$ be an arbitrary coset of $\mathcal{L}$ and $Y_i$ an independent sample from $D_{\mathcal{L} + \cv_i, s}$. 
    Then the distribution of $Y_1 - Y_2$ satisfies \begin{align*}
        \forall\, \yv \in \mathcal{L} + \cv_1 - \cv_2,\, \Pr[Y_1 - Y_2 = \yv] \in \left[\frac{1-\eps}{1+\eps}, \frac{1+\eps}{1-\eps} \right] \cdot D_{\mathcal{L} + \cv_1 - \cv_2, \sqrt{2}s}(\yv). 
    \end{align*} 
    %In particular, if $\eps \leq \frac{1}{2}$, then the distribution of $Y_1 - Y_2$ is $2\eps$-close and $3\eps$-similar to $D_{\mathcal{L} + \cv_1 - \cv_2, \sqrt{2}s}$. 
\end{lemma}
% Note that~\cite[Lemma~2.14]{ALS21} only obtains $6\eps$-similarity.  

It follows, for instance, that $Y_1 - Y_2$ is $3\eps$-similar to $D_{\mathcal{L} + \cv_1 - \cv_2, \sqrt{2}s}$ whenever $\eps \leq \frac{1}{2}$. %(and within statistical distance $2\eps$ from) 
% Proof: 
% If $0 < \eps \leq \frac{1}{2}$, then $\left[ \frac{1-\eps}{1+\eps}, \frac{1+\eps}{1-\eps} \right] \subseteq [1 - 4\eps, 1 + 4\eps]$, so the statistical distance is then at most $2\eps$. 
% In addition, $\left[ \frac{1-\eps}{1+\eps}, \frac{1+\eps}{1-\eps} \right] \subseteq [e^{-3\eps}, e^{3\eps}]$ for $0 \leq \eps \leq \frac{5}{6}$, so it is also $3\eps$-similar. (In fact, for $0 < \eps \leq \frac{1}{2}$, can show it is $2.2\eps$-similar.) 

\begin{proof} 
We write $D_1$ for $D_{\mathcal{L} + \cv_1, s}$ and $D_2$ for $D_{\mathcal{L} + \cv_2, s}$. The support of the distribution of $Y_1 - Y_2$, for $Y_1 \sim D_1$ and $Y_2 \sim D_2$, is  $\mathcal{L} + \cv_1 - \cv_2$. 
For all $\xv \in \mathcal{L} + \cv_1 - \cv_2$, we have \begin{align}
    \Pr_{\substack{Y_1 \sim D_1\\Y_2 \sim D_2}}[Y_1 - Y_2 = \xv] &= \sum_{\yv_1 \in \mathcal{L} + \cv_1} \Pr_{\substack{Y_1 \sim D_1\\Y_2 \sim D_2}}[Y_1 = \yv_1 \text{ and }Y_1 - Y_2 = \xv] \notag \\
    &= \sum_{\yv_1 \in \mathcal{L} + \cv_1} \Pr_{Y_1 \sim D_1}[Y_1 = \yv_1 ] \cdot \Pr_{\substack{Y_1 \sim D_1\\ Y_2 \sim D_2}}[Y_1 - Y_2 = \xv \mid Y_1 = \yv_1 ] \notag \\ 
    &= \sum_{\yv_1 \in \mathcal{L} + \cv_1} \frac{\rho_s(\yv_1)}{\rho_s(\mathcal{L} + \cv_1)} \cdot \frac{\rho_s(\yv_1 - \xv)}{\rho_s(\mathcal{L} + \cv_2)} \tag{def. of $D_1,D_2$} \\
    &= \rho_{\sqrt{2}s}(\xv) \cdot \sum_{\yv_1 \in \mathcal{L} + \cv_1} \frac{\rho_{s/\sqrt{2}}(\yv_1 - \xv/2)}{\rho_s(\mathcal{L} + \cv_1) \cdot \rho_s(\mathcal{L} + \cv_2)} \label{eq: rewriting in conv lemma} \\
    &= \rho_{\sqrt{2}s}(\xv) \cdot \frac{\rho_{s/\sqrt{2}}(\mathcal{L} + \cv_1 - \xv/2)}{\rho_s(\mathcal{L} + \cv_1) \cdot \rho_s(\mathcal{L} + \cv_2)} \notag
\end{align}
where \Cref{eq: rewriting in conv lemma} holds since  $\rho_s(\vv_1) \cdot \rho_s(\vv_1 - \vv_2) = \rho_{s}(\vv_2/\sqrt{2}) \cdot \rho_{s}(\sqrt{2}\vv_1 - \vv_2/\sqrt{2}) = \rho_{\sqrt{2}s}(\vv_2) \cdot \rho_{s/\sqrt{2}}(\vv_1 - \vv_2/2)$ for all $\vv_1, \vv_2 \in \R^n$. 
Since $s \geq \sqrt{2}\eta_\eps(\mathcal{L})$, \Cref{lem: Lemma 4.4 in MR2004 implicit version} implies $\rho_{s/\sqrt{2}}(\mathcal{L} + \cv_1 - \xv/2) \in \left[\frac{1-\eps}{1+\eps}, 1 \right] \cdot \rho_{s/\sqrt{2}}(\mathcal{L})$. 
Hence, \begin{align*}
    \Pr_{\substack{Y_1 \sim D_1\\ Y_2 \sim D_2}}[Y_1 - Y_2 = \xv] \in \left[\frac{1-\eps}{1+\eps}, 1 \right] \cdot  \rho_{\sqrt{2}s}(\xv) \cdot \frac{\rho_{s/\sqrt{2}}(\mathcal{L})}{\rho_s(\mathcal{L} + \cv_1) \cdot \rho_s(\mathcal{L} + \cv_2)}. 
\end{align*} 
Summing both sides implies that $1 \in \left[\frac{1-\eps}{1+\eps}, 1 \right] \cdot  \rho_{\sqrt{2}s}(\mathcal{L} + \cv_1 - \cv_2) \cdot \frac{\rho_{s/\sqrt{2}}(\mathcal{L})}{\rho_s(\mathcal{L} + \cv_1) \cdot \rho_s(\mathcal{L} + \cv_2)}$, i.e., $\frac{\rho_{s/\sqrt{2}}(\mathcal{L})}{\rho_s(\mathcal{L} + \cv_1) \cdot \rho_s(\mathcal{L} + \cv_2)} \in \left[1, \frac{1+\eps}{1-\eps} \right] \cdot \frac{1}{\rho_{\sqrt{2}s}(\mathcal{L} + \cv_1 - \cv_2)}$. 
It follows that \begin{align*}
    \Pr_{\substack{Y_1 \sim D_1\\ Y_2 \sim D_2}}[Y_1 - Y_2 = \xv] \in \left[\frac{1-\eps}{1+\eps}, \frac{1+\eps}{1-\eps} \right] \cdot \frac{\rho_{\sqrt{2}s}(\xv)}{\rho_{\sqrt{2}s}(\mathcal{L} + \cv_1 - \cv_2)}
\end{align*} as we wanted to show. \qed 
\end{proof}

\subsection{Bounds on Smoothing Parameters of Relevant Lattices}\label{sec: relevant smoothing bounds}

The following lemma gives a bound on the smoothing parameter of $\Z^n$. (Note that it can also be viewed as a special case of \Cref{lem: lemma 3.5 in Pei2008} below.) 

\begin{lemma}[Special Case of {\cite[Lemma~3.3]{MR2004}}]\label{lem: lemma 3.3 in MR2004 applied to Zn}
For any real $\eps > 0$, 
\begin{align*}
    \eta_\eps(\Z^n) \leq \sqrt{\frac{\ln(2n(1 + 1/\eps))}{\pi}}.
\end{align*}
\end{lemma}

The next lemma gives an upper bound on $\eta_\eps(\mathcal{L})$ in terms of $\lambda_1^\infty(\mathcal{L}^*)$. 
It will be used to obtain an upper bound on the smoothing parameter of the lattices $\Lambda_i$ that we will consider. 

% Proof of this lemma follows from Ban95 Lemma 2.10 (with natural logarithm). 
% Note that the lemma is stated differently (seems a typo?) in GPV2008 (Lemma 2.6). 
\begin{lemma}[Part of {\cite[Lemma~3.5]{Pei2008}}]\label{lem: lemma 3.5 in Pei2008}
    For any full-rank lattice $\mathcal{L} \subseteq \R^n$ and real $\eps > 0$, 
    \begin{align*}
        \lambda_1^\infty(\mathcal{L}^*) \cdot \eta_\eps(\mathcal{L}) \leq \sqrt{\frac{\ln(2n(1 + 1/\eps))}{\pi}}. 
    \end{align*}
\end{lemma}

The lattices $\Lambda_i$ considered in this work are of the form $\Lambda_i = \Lambda_q^\bot(\Am_i)$ for some matrices $\Am_i$, and we recall that their dual lattice is of the form $\frac{1}{q} \Lambda_q(\Am_i)$. We obtain the following lower bound on $\lambda_1^{\infty}(\Lambda_q(\Am))$ for random matrices $\Am \in \Z_q^{n \times m}$.

\begin{lemma}[Variant of {\cite[Lemma~5.3]{GPV2008}}]\label{lem: lower bound on infinity lambda1 for q-ary}
    Let $n,m,q$ be positive integers with $q$ prime and $q^{1-n/m} \geq 6$. 
    For uniformly random $\Am \in \Z_q^{n \times m}$, we have that \begin{align*}
        \lambda_1^{\infty}(\Lambda_q(\Am))> \frac{q^{1 - n/m} \cdot 2^{-n/m}}{3}
    \end{align*}
    except with probability $< 2^{-n}$. 
    
    In particular, if $m \geq n$, then the right-hand side is lower bounded by $\frac{q^{1 - n/m}}{6}$.
\end{lemma}

\begin{proof}
(In this proof, the probabilities are taken over all uniformly random $\Am \in \Z_q^{n \times m}$.)
For some positive real $B$ to be determined, let $S := \{\yv \in \Z^m \colon \|\yv\|_\infty \leq B\}$, and note that $|S| = (2B+1)^m$. 
Furthermore, for all $\sv \in \Z_q^n \setminus \{\zerovec\}$, $\Pr[\Am^{\top} \sv \bmod{q} \in S]  = |S| \cdot q^{-m} = (2B+1)^m \cdot q^{-m}$. 
Taking the union bound over all $\sv \in \Z_q^n \setminus \{\zerovec\}$ gives \begin{align*}
    \Pr[\lambda_1^{\infty}(\Lambda_q(\Am)) \leq B] &= \Pr[\exists \sv \in \Z_q^n \setminus \{\zerovec\} \text{ such that } \Am^{\top} \sv \bmod{q} \in S] \\
    &\leq |\Z_q^n \setminus \{\zerovec\}| \cdot (2B+1)^m \cdot q^{-m} \\
    &< (2B+1)^m \cdot q^{n-m}.
\end{align*}
Let $B:= \frac{1}{3}q^{1 - n/m} \cdot 2^{-n/m}$, and observe that $q^{1-n/m} \geq 6$ implies $q^{1-n/m} \cdot 2^{-n/m} \geq 3$ for all $m\geq n$. It follows that $B \geq 1$ and thus $\Pr[\lambda_1^{\infty}(\Lambda_q(\Am)) \leq B] < (3B)^m \cdot q^{n-m} = 2^{-n}$ as desired. The last part is immediate. 
\hfill $\qed$ \end{proof}

\begin{lemma}[Smoothing Parameter of $\Lambda_q^\bot(\Am)$]\label{lem: smoothing parameter of q-ary kernel lattice}
Let $n,m,q$ be positive integers with $q$ prime, $m \geq n$, and $q^{1-n/m} \geq 6$. 
Let $\eps \leq \frac{1}{4m}$ be a positive real. 
For uniformly random $\Am \in \Z_q^{n \times m}$, we have that 
\begin{align*}
    \eta_\eps(\Lambda_q^\bot(\Am)) < \sqrt{\frac{72\ln(1/\eps)}{\pi}} \cdot q^{n/m}
\end{align*}
except with probability $< 2^{-n}$.
\end{lemma}

\begin{proof} 
    Since $m \geq n$,  Lemma~\ref{lem: lower bound on infinity lambda1 for q-ary} for a uniformly random $\Am \in \Z_q^{n \times m}$ implies that $\lambda_1^\infty(\Lambda_q(\Am)) > \frac{1}{6}q^{1 - n/m}$, except with probability $< 2^{-n}$. 
    Since the dual of $\mathcal{L} := \Lambda_q^\bot(\Am)$ is $\mathcal{L}^* = \tfrac{1}{q} \Lambda_q(\Am)$, we obtain that $\lambda_1^\infty(\mathcal{L}^*) > \frac{1}{6} q^{- n/m}$.  
    Furthermore, by \Cref{lem: lemma 3.5 in Pei2008} (recall that $\mathcal{L}$ is full-rank and has dimension $m$), we have \begin{align*}
        \lambda_1^\infty(\mathcal{L}^*) \cdot \eta_\eps(\mathcal{L}) \leq \sqrt{\frac{\ln(2m(1 + 1/\eps))}{\pi}}
    \end{align*}
    which is $\leq \sqrt{\frac{2\ln(1/\eps)}{\pi}}$ for $\eps \leq \frac{1}{4m}$. %(Here, we use that $\ln(2m(1 + 1/\eps)) = \ln(2m) + \ln(1 + \eps) + \ln(1/\eps) \leq \ln(4m) + \ln(1/\eps)$.) % NB: We can also consider bounding by $\leq \sqrt{\frac{2\ln(4/\eps)}{\pi}}$ for $\eps \leq \frac{4}{m}$. 
    The statement follows. 
\hfill $\qed$ \end{proof}

\section{Wagner-Style Gaussian Sampler} \label{sec: Wagner Gaussian sampler (full section)}

In \Cref{sec: technical introduction}, we presented a warm-up version of the Wagner-style algorithm, which returns $N$ short vectors in $\Lambda_q^\bot(\Am)$. However, these vectors are possibly all equal to $\zerovec$. We will now show that a variant of that algorithm, using discrete Gaussians, allows us to avoid that issue. 
In particular, we present a Wagner-style algorithm for sampling $N$ vectors from a distribution that is essentially $D_{\Lambda_q^\bot(\Am), s}$ when $s$ is sufficiently large. Such samples can be shown to be short and nonzero with high probability. %\footnote{Note that the parameter $s$ allows to control the $\ell_\infty$-norm of the samples from  $D_{\Lambda_q^\bot(\Am), s}$ (by \Cref{lem: Lemma 2.10 in Ban95}).} 
Specifically, we present an algorithm for sampling from $D_{\Lambda_q^\bot(\Am), s}$ in time subexponential in $n$ for $m = n + \omega(n/\log \log n) $, $q = \mathrm{poly}(n)$, and $s = q/f$ for some $f = \omega(1)$. % NB: "for some" because f cannot be too large 

Recall that, for some $r \in \N$ and $\nnew_1,\ldots,\nnew_r$ such that $n = \sum_{i=1}^r \nnew_i$, we define the $q$-ary lattices
\begin{align}\label{eq: def of Lambda}
    \Lambda_0 = \Z^{m-n}  
    \quad \text{ and } \quad
    \Lambda_i = \Lambda_q^\bot(\Am_i) = \{\xv \in \Z^{m-n + \nsum_i} \colon \Am_i \xv = \zerovec \bmod{q}\}
\end{align}
for $i = 1, \ldots, r$, where $\Am_i \in \Z_q^{\nsum_i \times (m - n + \nsum_i)}$ is the matrix corresponding to the first $\nsum_i := \sum_{j=1}^i \nnew_j$ SIS equations. (Recall \Cref{fig: matrix}.) %and $\Am_r = \Am$.
%As before, each $\Am_i$ can be written as $\Am_i = [\Am'_i \mid \Id_{\nsum_i}]$ for $\Am'_i \in \Z_q^{\nsum_i \times (m - n + \nsum_i)}$. 
In other words, our goal is to sample $N$ vectors from $D_{\Lambda_r, s}$ for a given parameter $s$.

Our approach is to start from many vectors sampled from $D_{\Lambda_0, s_0}$, where $s_0$ is such that $s = \sqrt{2^r} s_0$. Then, we iteratively (for $i \in \{1,\ldots,r\}$) transform a list of vectors that are conditionally similar to independent samples from $D_{\Lambda_{i-1}, \sqrt{2^{i-1}} s_0}$ into a list of samples that are conditionally similar to independent samples from $D_{\Lambda_{i}, \sqrt{2^i} s_0}$. Then, after the last iteration, the list contains samples that are conditionally similar to independent samples from $D_{\Lambda_q^\bot(\Am), s}$, as desired. (Using \Cref{lem: upper bound on maximum discrete Gaussian probability (min-entropy)}, we can then bound the probability that one such sample is nonzero.)

As explained in \Cref{sec: technical introduction}, the mapping from vectors in $\Lambda_{i-1}$ to vectors in $\Lambda_i$ will be done by first lifting the vectors to vectors in a suitable superlattice $\Lambda'_i \supseteq \Lambda_i$, and then combining them into vectors in $\Lambda_i$. 
Specifically, for some $p_i \in \N$, the lattices $\Lambda'_i$ are defined by $\Lambda'_i = \mathcal{L}(\Bm'_i)$ for 
\begin{align} \label{eq: basis B'i}
    \Bm'_i = \begin{pmatrix}
        \zerovec & \Id_{m-n} \\
        \Dm_i & -\Am'_i
    \end{pmatrix} \text{ with }
    \Dm_i := \begin{pmatrix}
        \zerovec & q\Id_{\nsum_{i-1}}  \\  \frac{q}{p_i} \Id_{\nnew_i} & \zerovec
    \end{pmatrix}
\end{align} 

The first $\nnew_i$ columns of $\Bm'_i$ generate (an embedding into $\R^{m - n + \nsum_i}$ of) $\frac{q}{p_i}\Z^{\nnew_i}$, which is a primitive sublattice of $\Lambda'_i$ that we denote by $\mathcal S$. Consider the projected lattice $\mathcal P = \pi_{\mathcal{S}}^\bot(\Lambda'_i)$, and note that it is (an embedding of) $\Lambda_{i-1}$. Thus, we can consider ways to lift from $\Lambda_{i-1}$ to $\Lambda'_i$; in \Cref{sec: discrete-Gaussian-preserved lifting}, we will consider a randomized way of lifting that preserves discrete Gaussian distributions. 

From there, we would like to produce samples in $\Lambda_i$, rather than $\Lambda'_i$. A natural approach is to bucket our samples according to their cosets in the quotient $\Lambda'_i/\Lambda_i$, and then take differences within those buckets. Using standard analysis of convolutions of discrete Gaussians, we show in \Cref{sec: combining to a sublattice} that these differences are still essentially discrete Gaussian, yet with a width parameter increased by a factor of $\sqrt 2$.

In \Cref{sec: Wagner as Gaussian sampler algorithm}, we then lay out the resulting Gaussian variant of Wagner's algorithm, and demonstrate its correctness and time complexity, under certain smoothing constraints on the parameters. Finally, in \Cref{sec: complexity analysis wrap-up} we provide the choice of parameters that satisfy those constraints and conclude with a subexponential-time algorithm for sampling from $D_{\Lambda_q^\bot(\Am), s}$.

\subsection{Discrete-Gaussian Lifting}\label{sec: discrete-Gaussian-preserved lifting}
% In this section: under suitable conditions, we can transform vectors in $\Lambda_{i-1}$ to vectors in $\Lambda'_i$ while preserving similarity

\Cref{alg: DGLift} below lifts vectors from $\Lambda_{i-1}$  to vectors in $\Lambda'_i$. It revisits the GPV sampling algorithm~\cite{GPV2008} with a reinterpretation of the induction: rather than reducing the problem in dimension $n$ to two instances in dimensions $1$ and $n-1$, we consider arbitrary splits in $n'$ and $n-n'$ dimensions. 
\Cref{alg: DGLift} can be viewed as a special case of~\cite[Alg.~2]{EWY23}, and our \Cref{lem: conditionally similar DGLift} is a variant of~\cite[Theorem~1]{EWY23}, where we analyze the conditional similarity of the output with respect to the discrete Gaussian (instead of merely looking at the statistical distance). 

In particular, \Cref{lem: conditionally similar DGLift} (with $\delta = 0$ and $N=1$) shows that \Cref{alg: DGLift} turns a sample from $D_{\mathcal{P}, s}$ into a sample that is $3\eps$-similar to (and thus within statistical distance $3\eps$ from) $D_{\mathcal{L}, s}$, whenever $s \geq \eta_\eps(\mathcal{S})$ for $0 < \eps \leq \frac{1}{2}$. % This also can be seen to follow from the proof of~\cite[Theorem~1]{EWY23}.
% NB: It actually can be seen from our proof that we obtain statistical distance $2\eps$. 

% Revisiting GPV sampling 
\begin{lemma}[Complexity and Distribution of $\mathrm{DGLift}$]\label{lem: conditionally similar DGLift} 
    Let $\mathcal{L} \subseteq \R^n$ be a lattice and let $\mathcal{P} = \pi_{\mathcal{S}}^{\bot}(\mathcal{L})$ for a primitive sublattice $\mathcal{S} \subseteq \mathcal{L}$. 
    Let $s > 0$ be a real such that a randomized algorithm $\mathcal{A}$ exists that, given $\cv \in \Span(\mathcal{S})$, returns a sample from $D_{\mathcal{S}, s, \cv}$. Then $\mathrm{DGLift}(\mathcal{P}, \mathcal{L}, s, \cdot)$ (\Cref{alg: DGLift}) is a randomized algorithm that, given a vector $\xv \in \mathcal{P}$, outputs a vector $\xv'$ in $\mathcal{L}$. 
    It uses one query to $\mathcal{A}$, and all other operations run in polynomial time. 
    
    % Moreover, for any real $0 < \eps \leq \frac{1}{2}$, if $s \geq \eta_\eps(\mathcal{S})$ and $\xv \sim D_{\mathcal{P}, s}$, then the output distribution is $2\eps$-close and $3\eps$-similar to $D_{\mathcal{L},s}$. 

    Moreover, for any reals $\delta \geq 0$ and $0 < \eps \leq \frac{1}{2}$ satisfying $s \geq \eta_\eps(\mathcal{S})$, if $X_1, \ldots, X_N \in \mathcal{P}$ are conditionally $\delta$-similar to independent samples from $D_{\mathcal{P}, s}$, then the distribution of $X'_1, \ldots, X'_N$ for $X'_i := \mathrm{DGLift}(\mathcal{P}, \mathcal{L}, s, X_i)$ is conditionally $(\delta + 3\eps)$-similar to independent samples from $D_{\mathcal{L}, s}$. 
\end{lemma}

We will use the above lemma with $\mathcal S = \frac{q}{p_i} \mathbb Z^{\nnew_i}$, hence an exact polynomial-time sampler is available whenever $s \geq \frac{q}{p_i} \sqrt{\ln(2 \nnew_i+4)/\pi}$ by~\Cref{lem: exact randomized Gaussian rounding}.

\begin{algorithm}[H]
\caption{$\mathrm{DGLift}(\mathcal{P}, \mathcal{L}, s, \xv)$}
\label{alg: DGLift}
\DontPrintSemicolon
\SetKwInOut{Input}{Input}\SetKwInOut{Output}{Output}
\Input{Lattices $\mathcal P, \mathcal{L}$, where $\mathcal P = \pi_{\mathcal S}^\bot(\mathcal{L})$ for a primitive sublattice $\mathcal S \subseteq \mathcal L$; \\ 
Real $s > 0$; \\
Vector $\xv \in \mathcal{P}$}
\Output{Vector $\xv' \in \mathcal{L}$ such that $\pi_{\mathcal{S}}^{\bot}(\xv') = \xv$}
\vspace{2mm}

Let $\mathcal{C}$ be a complement to $\mathcal{S}$ \;
Compute the unique $\yv \in \mathcal{C}$ such that $\pi_{\mathcal{S}}^{\bot}(\yv) = \xv$ \tcp*{Lifting }  
Sample $\zv \sim D_{\mathcal{S}, s, \xv - \yv}$ \tcp*{Sampling} 
\Return{$\xv' := \zv + \yv$}
\end{algorithm}
 
\begin{proof}
Consider \Cref{alg: DGLift}, where we use algorithm $\mathcal{A}$ for the sampling step, and let $\mathcal{C}$ be the complement of $\mathcal{S}$ that it considers. Note that $\pi_{S}^{\bot}$ induces a bijection from $\mathcal{C}$ to $\mathcal{P}$, and that both directions can be computed in polynomial time. 
The claim on time and query complexity of \Cref{alg: DGLift} is thus immediate. 

For correctness, we remark that any output vector $\xv' =  \zv + \yv$ belongs to $\mathcal{L}$, since $\zv \in \mathcal{S} \subseteq \mathcal{L}$ and  $\yv \in \mathcal{C} \subseteq \mathcal{L}$. Furthermore, $\xv' = \pi_{\mathcal{S}}(\xv') + \pi_{\mathcal{S}}^{\bot}(\xv')$ with $\pi_{\mathcal{S}}(\xv') = \zv + \pi_{\mathcal{S}}(\yv)$ and $\pi_{\mathcal{S}}^{\bot}(\xv') = \pi_{\mathcal{S}}^{\bot}(\yv) = \xv$, so the output is as desired. 
% NB: x-y in span(S), so sampling is also well-defined 

For the remainder of the proof, let $f(\xv) := \mathrm{DGLift}(\mathcal{P}, \mathcal{L}, s, \xv)$. We remark that $\mathcal{S} \oplus \mathcal{C} = \mathcal{L}$ implies that any $\vv \in \mathcal{L}$ can be uniquely written as $\vv = \vv_\mathcal{S} + \vv_\mathcal{C}$ for $\vv_\mathcal{S} \in \mathcal{S}$ and $\vv_\mathcal{C} \in \mathcal{C}$ (and we will define $\vv_\mathcal{S}, \vv_\mathcal{C}$ as such). 

We first observe that for all $\xv \in \mathcal{P}$ and $\xv' \in \mathcal{L}$, the probability that $\mathrm{DGLift}$ on input $\xv$ outputs $\xv'$ is \begin{align}
    \Pr[f(\xv) = \xv'] &= \Pr_{Z \sim D_{\mathcal{S}, s, \xv - \yv(\xv)}}[Z = \xv' - \yv(\xv)] \notag \\
    &= \begin{cases}
        \frac{\rho_s(\pi_{\mathcal S}(\xv'))}{\rho_s(\mathcal{S} + \pi_{\mathcal{S}}(\xv'_{\mathcal{C}}))} &\text{if $\xv = \pi_{\mathcal{S}}^\perp(\xv')$} \\ 
        0 &\text{otherwise} 
    \end{cases} \label{eq: distribution of DGLift output for any input}
\end{align}
where $\yv(\xv)$ denotes the unique $\yv \in \mathcal{C}$ such that $\pi_{\mathcal{S}}^{\bot}(\yv) = \xv$. 
Indeed, note that $\xv' - \yv(\xv) \in \mathcal{S}$ if and only if $\xv'_{\mathcal{C}} = \yv(\xv)$ if and only if $\xv = \pi_{\mathcal{S}}^\perp(\xv'_{\mathcal{C}})$ if and only if $\xv = \pi_{\mathcal{S}}^\perp(\xv')$. 
Therefore, $\Pr_{Z \sim D_{\mathcal{S}, s, \xv - \yv(\xv)}}[Z = \xv' - \yv(\xv)] = \frac{\rho_s(\xv' - \xv)}{\rho_s(\mathcal{S} - \xv + \yv(\xv))} = \frac{\rho_s(\pi_{\mathcal S}(\xv'))}{\rho_s(\mathcal{S} - \xv + \yv(\xv))}$ if $\xv = \pi_{\mathcal{S}}^\perp(\xv')$ and 0 otherwise. 
Since $\xv = \pi_{\mathcal{S}}^\perp(\xv')$ implies that $\yv(\xv) = \xv'_{\mathcal{C}}$, and thus $\yv(\xv) - \xv = \pi_{\mathcal{S}}(\yv(\xv)) = \pi_{\mathcal{S}}(\xv'_{\mathcal{C}})$, \Cref{eq: distribution of DGLift output for any input} follows. 
% NB: The output distribution of $\mathrm{DGLift}(\xv)$ is thus $D_{\mathcal S + (\pi_{\mathcal{S}}^\perp)^{-1}(\xv), s, \xv} = D_{\mathcal S + \yv(\xv), s, \xv}$. 

Next, we prove the following intermediate claim. 
\begin{claim}\label{claim: DGLift lemma}
    For $s \geq \eta_\eps(\mathcal{S})$, we have $\rho_s(\mathcal{P})\cdot \rho_s(\mathcal{S}) \in \left[ 1, \frac{1+\eps}{1-\eps} \right] \cdot \rho_s(\mathcal{L})$. 
\end{claim}
\begin{proof}[of Claim]
By \Cref{eq: distribution of DGLift output for any input}, if the input is a random variable $X \sim D$ for some distribution $D$ on $\mathcal{P}$, then for all $\xv' \in \mathcal{L}$, we have that \begin{align*}
    \Pr_{X, Z}[f(X) = \xv'] &= \Pr_{X, Z}[X = \pi_{\mathcal{S}}^\perp(\xv')]  \cdot \Pr_{X, Z}[f(X) = \xv' \mid X = \pi_{\mathcal{S}}^\perp(\xv')] \\ 
    &= \Pr_{X \sim D}[X = \pi_{\mathcal{S}}^\perp(\xv')] \cdot \frac{\rho_s(\pi_{\mathcal{S}}(\xv'))}{\rho_s(\mathcal{S} + \pi_{\mathcal{S}}(\xv'_{\mathcal{C}}))} \\
    &\in \left[1, \frac{1 + \eps}{1 - \eps}\right] \cdot  \Pr_{X \sim D}[X = \pi_{\mathcal{S}}^\perp(\xv')] \cdot \frac{\rho_s(\pi_{\mathcal{S}}(\xv'))}{\rho_s(\mathcal{S})}. \tag{\Cref{lem: Lemma 4.4 in MR2004 implicit version}}
\end{align*} 
% The last line follows from \Cref{lem: Lemma 4.4 in MR2004 implicit version} (which we can apply since $s \geq \eta_\eps(\mathcal{S})$). 
In particular, if $D$ is $D_{\mathcal{P},s}$, then we have (for all $\xv' \in \mathcal{L}$) that $\Pr[f(X) = \xv'] \in \left[1, \frac{1 + \eps}{1 - \eps}\right] \cdot \frac{\rho_s(\xv')}{\rho_s(\mathcal{P})\rho_s(\mathcal{S})}$
since $\pi_{\mathcal{S}}^\perp(\xv') + \pi_{\mathcal{S}}(\xv') = \xv'$.
Finally, summing both sides over all $\xv' \in \mathcal{L}$ yields that $1 \in \left[1, \frac{1 + \eps}{1 - \eps}\right] \cdot \frac{\rho_s(\mathcal{L})}{\rho_s(\mathcal{P})\rho_s(\mathcal{S})}$, 
which proves the claim.  \qed 
\end{proof}

We now proceed with the main proof. Suppose that the input consists of $N$ random variables conditionally $\delta$-similar to independent samples from $D_{\mathcal{P},s}$.  
By \Cref{eq: distribution of DGLift output for any input}, we know that for all $\xv' \in \mathcal{L}^N$ and any $I \subseteq [N]$, \begin{align}
    &\, \Pr_{(X_1,Z_1),\ldots, (X_N, Z_N)}[ \forall j \in I, f(X_j) = \xv'_j] \notag \\
    =&\, \Pr_{X_1,\ldots, X_N}[\forall j \in I, X_j = \pi_{\mathcal{S}}^\perp(\xv'_j)]  \cdot \Pr_{Z_1,\ldots, Z_N}[\forall j \in I, f(\pi_{\mathcal{S}}^\perp(\xv'_j)) = \xv'_j] \notag \\ 
    =&\, \Pr_{X_1,\ldots, X_N}[\forall j \in I, X_j = \pi_{\mathcal{S}}^\perp(\xv'_j)] \cdot \prod_{j \in I} \frac{\rho_s(\pi_{\mathcal{S}}(\xv'_j))}{\rho_s(\mathcal{S} + \pi_{\mathcal{S}}(\xv'_{j,\mathcal{C}}))} \label{eq: DGLift lemma decomposition for independent Z}
\end{align} since the $Z_j$ are independent when the values of the $X_j$ are fixed. (Here, we write $\xv'_{j,\mathcal{C}}$ for the unique $\cv \in \mathcal{C}$ such that $\xv'_{j} = \sv + \cv$ for $(\sv,\cv) \in \mathcal{S}\times \mathcal{C}$.)
To conclude the proof, take any $i \in [N]$ and $\xv' \in \mathcal{L}^N$. Then 
{\allowdisplaybreaks
\begin{align*}
    &\, \Pr_{X_1,\ldots,X_N}[f(X_i) = \xv'_i \mid \forall j\in [N]\setminus\{i\},  f(X_j) = \xv'_j]  \\
    =&\, \frac{\Pr_{X_1,\ldots,X_N}[\forall j\in [N] ,  f(X_j) = \xv'_j]}{\Pr_{X_1,\ldots,X_N}[\forall j\in [N]\setminus\{i\},  f(X_j) = \xv'_j] } \tag{def. conditional probability} \\
    =&\, \frac{\Pr_{X_1,\ldots, X_N}[\forall j \in [N], X_j = \pi_{\mathcal{S}}^\perp(\xv'_j)] }{\Pr_{X_1,\ldots, X_N}[\forall j \in [N]\setminus\{i\}, X_j = \pi_{\mathcal{S}}^\perp(\xv'_j)] } \cdot \frac{\rho_s(\pi_{\mathcal{S}}(\xv'_i))}{\rho_s(\mathcal{S} + \pi_{\mathcal{S}}(\xv'_{i,\mathcal{C}}))} \tag{\Cref{eq: DGLift lemma decomposition for independent Z}} \\ 
    =&\, \Pr_{X_1,\ldots, X_N}[X_i = \pi_{\mathcal{S}}^\perp(\xv'_i)  \mid \forall j \in [N]\setminus\{i\}, X_j = \pi_{\mathcal{S}}^\perp(\xv'_j)]  \cdot \frac{\rho_s(\pi_{\mathcal{S}}(\xv'_i))}{\rho_s(\mathcal{S} + \pi_{\mathcal{S}}(\xv'_{i,\mathcal{C}}))}  \\ 
    =&\, e^{\pm \delta} \cdot \frac{\rho_s(\pi_{\mathcal{S}}^\perp(\xv'_i))}{\rho_s(\mathcal{P})}  \cdot \frac{\rho_s(\pi_{\mathcal{S}}(\xv'_i))}{\rho_s(\mathcal{S} + \pi_{\mathcal{S}}(\xv'_{i,\mathcal{C}}))} \tag{by assumption} \\ 
    =&\, e^{\pm \delta} \cdot \frac{\rho_s(\xv'_i)}{\rho_s(\mathcal{P}) \rho_s(\mathcal{S} + \pi_{\mathcal{S}}(\xv'_{i,\mathcal{C}}))}  \\ 
    \in&\, \left[e^{- \delta}, e^{\delta} \cdot \frac{1 + \eps}{1 - \eps}\right] \cdot \frac{\rho_s(\xv'_i)}{\rho_s(\mathcal{P}) \rho_s(\mathcal{S})} \tag{by \Cref{lem: Lemma 4.4 in MR2004 implicit version}} \\ 
    \subseteq&\, \left[e^{- \delta} \cdot \frac{1 - \eps}{1 + \eps}, e^{\delta} \cdot \frac{1 + \eps}{1 - \eps}\right]  \cdot \frac{\rho_s(\xv'_i)}{\rho_s(\mathcal{L})}.\tag{by the claim} 
\end{align*}} 

Since $\left[\frac{1 - \eps}{1 + \eps}, \frac{1 + \eps}{1 - \eps}\right] \subseteq [e^{- 3\eps}, e^{3\eps}]$ for all  $0 < \eps \leq \frac{1}{2}$, the lemma follows. \qed % NB: Holds even for all $0 < \eps \leq 5/6$.
\end{proof}

\subsection{Combining to a Sublattice}\label{sec: combining to a sublattice} % Combining Lattice Vectors to Vectors in a Full-Rank Sublattice

We will now show that, given many independent discrete Gaussian samples from a lattice $\mathcal{L}'$, we can construct many vectors in a \textit{full-rank} sublattice $\mathcal{L} \subseteq \mathcal{L}'$ whose distributions are (conditionally) similar to a discrete Gaussian over $\mathcal{L}$. 

By the convolution lemma~\cite{Pei2008,MP2013} (more precisely, by \Cref{lem: convolution lemma with specified distance}), the difference of two independent samples from $D_{\mathcal{L}',s}$ follows a distribution similar to $D_{\mathcal{L}', \sqrt{2}s}$. If we condition on the result being in the sublattice $\mathcal{L}$, then this distribution can in fact be shown to be similar to $D_{\mathcal{L}, \sqrt{2}s}$. 

Motivated by this fact, we consider an algorithm (\Cref{alg: bucket and combine}) that first buckets its input vectors in $\mathcal{L}'$ with respect to their cosets modulo the sublattice $\mathcal{L}$, and then (carefully) combines pairs of vectors in the same cosets to obtain vectors in $\mathcal{L}$.\footnote{\Cref{alg: bucket and combine} is just a reformulation of~\cite[Algorithm~2]{ALS21} with a different number of output vectors (and output vectors of the form $\xv - \xv'$ instead of $\xv + \xv'$).}  
% NB: Number of output vectors in ALS21 is $\lceil \frac{N - |\mathcal{L}' / \mathcal{L}|}{2} \rceil$ (assuming $N \geq 2 |\mathcal{L}' / \mathcal{L}|$). 
If we start with at least $3 |\mathcal{L}'/\mathcal{L}|$ vectors from $\mathcal{L}'$, then the number of output vectors is only a constant factor smaller, as shown by \Cref{lem: correctness of bucket and combine}. 
Furthermore, if the input vectors are conditionally similar to independent samples from $D_{\mathcal{L}',s}$, the output vectors are conditionally similar to independent samples from  $D_{\mathcal{L}, \sqrt{2}s}$, as shown by \Cref{lem: distribution of bucket and combine}.

\begin{algorithm}
\caption{$\mathrm{BucketAndCombine}(\mathcal{L}', \mathcal{L}, L)$} % Order is s.t. first lattice corresponds to input and second to output 
\label{alg: bucket and combine}
\DontPrintSemicolon
\SetKwInOut{Input}{Input}\SetKwInOut{Output}{Output}

\Input{Full-rank lattices $\mathcal{L} \subseteq \mathcal{L}'$ in $\R^d$; \\ 
A list $L$ with $N$ vectors $\xv_1, \ldots, \xv_{N} \in \mathcal{L}'$ for some integer $N \geq 3|\mathcal{L}'/\mathcal{L}|$}
\Output{A list $L_{out}$ with $\lfloor N/3\rfloor$ vectors in $\mathcal{L}$} 
\vspace{2mm}

Initialize empty lists $B(\cv)$ for each coset $\cv \in \mathcal{L}' / \mathcal{L}$ \; 
\For(\tcp*[f]{Bucketing}){$i = 1, \ldots, N$}{
    Let $\cv_i := \xv_i \bmod \mathcal{L}$ \;
    Append $\xv_i$ to $B(\cv_i)$ % i.e., add to the end of the list
}
Initialize an empty list $L_{out}$ \;
\For(\tcp*[f]{Combining}){$i = 1, \ldots, N$}{\If{$B(\cv_i)$ contains at least two elements and $|L_{out}| < \lfloor N/3\rfloor$}{Let $\xv, \xv'$ be the first two elements in $B(\cv_i)$ \; 
    Append $\yv := \xv - \xv'$ to $L_{out}$ \;
    Remove $\xv$ and $\xv'$ from $B(\cv_i)$ \;}
} 
\Return{$L_{out}$}
\end{algorithm}

\begin{lemma}[Correctness of \Cref{alg: bucket and combine}]\label{lem: correctness of bucket and combine}
    \Cref{alg: bucket and combine} is correct. That is, on input two full-rank lattices $\mathcal{L} \subseteq \mathcal{L}'$ in $\R^d$ and a list of $N$ vectors in $\mathcal{L}'$, it returns a list of $\lfloor N/3\rfloor$ vectors in $\mathcal{L}$ if $N \geq 3|\mathcal{L}'/\mathcal{L}|$. 
\end{lemma}

\begin{proof}
By construction, each element of $L_{out}$ is of the form $\yv = \xv - \xv'$ for $\xv = \xv' \bmod{\mathcal{L}}$, so $\yv \in \mathcal{L}$. It thus remains to show that the output list $L_{out}$ always consists of $\lfloor N/3\rfloor$ elements. 
Suppose, for contradiction, that the algorithm returns a list of size $\ell$ for some $\ell \in \{0, \ldots, \lfloor N/3\rfloor-1\}$. Then the number of used $L$ elements (i.e., those $\xv_i$ that are being used as part of the output) is $2\ell$, so there must be $N - 2\ell$ list elements that are not used. (Here, we talk about list elements instead of vectors, since two list elements may correspond to the same vector.) Their corresponding cosets must be distinct (since otherwise the algorithm would have been able to find more than $\ell$ output elements), so $N - 2\ell \leq |\mathcal{L}'/\mathcal{L}| \leq N/3$. It follows that $\lfloor N/3\rfloor \leq N/3 \leq \ell$, which is a contradiction. 
Hence, the algorithm always succeeds to construct $\lfloor N/3\rfloor$ output vectors. \qed 
\end{proof}

The following is a variant of~\cite[Lemma~4.5]{ALS21}, suitable for our purposes.  

\begin{lemma}[Distribution of Output]\label{lem: distribution of bucket and combine}
    Let $\mathcal{L} \subseteq \mathcal{L}'$ be full-rank lattices in $\R^d$. Let $N \geq 3|\mathcal{L}'/\mathcal{L}|$ be a positive integer, and let $\delta \geq 0$, $0 < \eps \leq \frac{1}{2}$, and $s \geq \sqrt{2}\eta_\eps(\mathcal{L}')$ be reals. 

    If the input list consists of $N$ random variables on $\mathcal{L}'$ that are conditionally $\delta$-similar to independent samples from $D_{\mathcal{L}',s}$, then \Cref{alg: bucket and combine} returns a list of $\lfloor N/3 \rfloor$ vectors from $\mathcal{L}$ that are conditionally $(4\delta + 3\eps)$-similar to independent samples from $D_{\mathcal{L},\sqrt{2}s}$.   
\end{lemma} 

Our proof makes use of the following fact: given a sample $X$ from a distribution similar to $D_{\mathcal{L}',s}$, if we condition on $X = \cv \bmod{\mathcal L}$ for some $\mathcal{L} \subseteq \mathcal{L}'$ and $\cv \in \mathcal L'$, then this distribution is similar to  $D_{\mathcal{L} + \cv,s}$. More generally, conditioning on cosets preserves conditional similarity. 

\begin{lemma}[Conditioning on Cosets]\label{lem: conditional similarity conditioned on cosets}
    Let $\mathcal{L} \subseteq \mathcal{L}'$ be full-rank lattices in $\R^d$. Let $N$ be a positive integer, and let $\delta \geq 0$ and $s > 0$ be reals. 
    Suppose that $X_1, \ldots, X_N \in \mathcal{L}'$ are discrete random variables that are conditionally $\delta$-similar to independent samples from $D_{\mathcal{L}', s}$. 
    Then, for all $i \in [N]$, all $\cv \in (\mathcal{L}'/\mathcal{L})^{N}$, and all $\xv \in (\mathcal{L}')^{N}$ satisfying $\xv_j = \cv_j \bmod{\mathcal{L}}$ for all $j \in [N]$,  
    \begin{align*}
    \Pr[X_i = \xv_i \mid X_{-i} = \xv_{-i} \text{ and } \forall j \in [N], X_j = \cv_j \bmod \mathcal{L}]  = e^{\pm 2\delta} \cdot \frac{\rho_s(\xv_i)}{\rho_s(\mathcal{L} + \cv_i)}.
\end{align*}

% In other words: ``Then, for all $\cv_1,\ldots,\cv_N \in  \mathcal{L}'/\mathcal{L}$, we have that $X_1, \ldots, X_N$ are conditionally $2\delta$-similar to independent samples from $D_{\mathcal{L} + \cv_1,s}, \ldots, D_{\mathcal{L} + \cv_N,s}$ (resp.).''
\end{lemma}

\begin{proof}
Suppose that $X = (X_1, \ldots, X_{N})$ consists of $N$ random variables on $\mathcal{L}'$ that are conditionally $\delta$-similar to independent samples from $D_{\mathcal{L}',s}$. 
Then, by definition, we have for all $i \in [N]$ and all $\xv \in (\mathcal{L}')^{N}$ that
\begin{align}
    \Pr[X_i = \xv_i \mid X_{-i} = \xv_{-i}] = e^{\pm \delta} \cdot \frac{\rho_s(\xv_i)}{\rho_s(\mathcal{L}')}. \label{eq: assumption for conditional similarity conditioned on cosets}
\end{align}
Consider arbitrary $i \in [N]$, $\cv \in (\mathcal{L}'/\mathcal{L})^{N}$, and $\xv \in (\mathcal{L}')^{N}$ satisfying $\xv_j = \cv_j \bmod{\mathcal{L}}$ for all $j \in [N]$. 
Then, by definition of conditional probability,
{\allowdisplaybreaks 
\begin{align*}
    &\, \Pr[X_i = \xv_i \mid X_{-i} = \xv_{-i} \text{ and } \forall j \in [N], X_j = \cv_j \bmod \mathcal{L}]  \\
    =&\, \frac{\Pr[X_i = \xv_i \text{ and } X_{-i} = \xv_{-i} \text{ and } \forall j \in [N], X_j = \cv_j \bmod \mathcal{L}]}{\Pr[X_{-i} = \xv_{-i} \text{ and } \forall j \in [N], X_j = \cv_j \bmod \mathcal{L}]} \\
    =&\, \frac{\Pr[X_i = \xv_i \text{ and } X_{-i} = \xv_{-i} ]}{\Pr[X_{-i} = \xv_{-i} \text{ and } X_i = \cv_i \bmod \mathcal{L}]} \tag{as $\xv_j = \cv_j \bmod{\mathcal{L}}$ $\forall j \in [N]$} \\
    =&\, \frac{\Pr[X_i = \xv_i \mid X_{-i} = \xv_{-i} ]}{\Pr[X_i = \cv_i \bmod \mathcal{L} \mid X_{-i} = \xv_{-i}]} \\ 
    =&\, \frac{\Pr[X_i = \xv_i \mid X_{-i} = \xv_{-i} ]}{ \sum_{\vv \in \mathcal{L} + \cv_i}\Pr[X_i = \vv \mid X_{-i} = \xv_{-i}]} \\
    =&\, e^{\pm 2\delta} \frac{{\rho_s(\xv_i)}/{\rho_s(\mathcal{L}')}}{{\rho_s(\mathcal{L} + \cv_i)}/{\rho_s(\mathcal{L}')}}
\end{align*}
where we apply \Cref{eq: assumption for conditional similarity conditioned on cosets} twice (to both the numerator and denominator) to obtain the last line. The conclusion then immediately follows.} \qed 
\end{proof} 

We can now prove \Cref{lem: distribution of bucket and combine}.

% Proof intuition: Once $\cv$ is fixed, the way the vectors are paired is entirely determined. From there, reorder the $X_i$ and $Y_i$ s.t. the algorithm performs $Y_i \gets X_{2i} - X_{2i + 1}$. 

\begin{proof}[of {\Cref{lem: distribution of bucket and combine}}]
Correctness (on arbitrary input) follows from \Cref{lem: correctness of bucket and combine}.

Let $Y_1,\ldots, Y_M$ be the random variables corresponding to the vectors in the output list (in order), where $M := \lfloor N/3 \rfloor$. 
We want to show that, for all $j \in [M]$ and $\yv \in \mathcal{L}^M$,
\begin{align*}
    \Pr_{X_1,\ldots,X_{N}}[Y_j = \yv_j \mid Y_{-j} = \yv_{-j}] &= e^{\pm (4\delta + 3\eps)} D_{\mathcal{L}, \sqrt{2}s}(\yv_j). 
\end{align*}
By \Cref{lem: similarity and conditional similarity closed under convex combinations}, it suffices to show that for all $\cv_1,\ldots, \cv_{N} \in \mathcal{L}'/\mathcal{L}$  such that $\Pr_{X_1,\ldots,X_{N}}[\forall i \in [N], X_i = \cv_i \bmod{\mathcal{L}} ] > 0$, we have, for all $j \in [M]$, $\yv \in \mathcal{L}^M$,  \begin{align}
    &\, \Pr_{X_1,\ldots,X_{N}}[Y_j = \yv_j \mid Y_{-j} = \yv_{-j} \text{ and } \forall i \in [N], X_i = \cv_i \bmod{\mathcal{L}}] \notag \\
    =&\, e^{\pm (4\delta + 3\eps)} D_{\mathcal{L}, \sqrt{2}s}(\yv_j). \label{eq: condition output vectors for conditional similarity}
\end{align}  

So consider any $\cv_1,\ldots, \cv_{N} \in \mathcal{L}'/\mathcal{L}$ such that $\Pr_{X_1,\ldots,X_{N}}[\forall i \in [N], X_i = \cv_i \bmod{\mathcal{L}}] > 0$. Note that the output (in particular, the way the vectors are paired) is entirely determined by the cosets $(\cv_1, \ldots, \cv_{N})$ for $\cv_i := \xv_i \bmod \mathcal{L}$.  
In particular, for any permutation $\pi \colon [N] \to [N]$ such that $(\xv'_1, \ldots, \xv'_{N}) := (\xv_{\pi(1)}, \ldots, \xv_{\pi(N)})$ satisfies $\yv_j = \xv'_{2j-1} - \xv'_{2j}$ for all $j \in [M]$, we have that $\xv'_{1}, \ldots, \xv'_{2M}$ is entirely determined by the cosets $(\cv_1, \ldots, \cv_{N})$.  
(The order of the vectors $\xv'_{2M+1}, \ldots, \xv'_{N}$ does not affect the algorithm's output.) % NB: N >= 2M+1, since N >= 3 
Without loss of generality, we will therefore redefine $(\xv_1,\ldots, \xv_{N})$ as $(\xv'_1, \ldots, \xv'_{N})$ for such a permutation (allowing us to write $\yv_j = \xv_{2j-1} - \xv_{2j}$ for all $j\in [M]$). 

Let $D$ be the distribution of the input variables $X = (X_1,\ldots, X_{N})$ \textit{conditional on} $X_i = \cv_i \bmod \mathcal{L}$ for all $i \in [N]$. 
By the conditional similarity assumption and by \Cref{lem: conditional similarity conditioned on cosets}, for all $i \in [N]$ and all $\xv \in (\mathcal{L}')^{N}$ satisfying $\xv_j =  \cv_j \bmod \mathcal{L}$ for all $j \in [N]$, we have % all $\cv \in (\mathcal{L}'/\mathcal{L})^{N}$
\begin{align}
    \Pr_{X \sim D}[X_i = \xv_i \mid X_{-i} = \xv_{-i}]  = e^{\pm 2\delta} \cdot D_{\mathcal{L} + \cv_i, s}(\xv_i), \label{eq: output distribution bucket and combine conditioned on cosets}
\end{align}
which we will repeatedly use below. 

To show that \Cref{eq: condition output vectors for conditional similarity} holds (for any $j \in [M]$, $\yv \in \mathcal{L}^M$), and thus to finish the proof, it suffices (by another application of \Cref{lem: similarity and conditional similarity closed under convex combinations}) to show that for  all $j \in [M]$, $\yv \in \mathcal{L}^M$, and all $\vv \in (\mathcal{L}')^{N-2}$ satisfying $\Pr[X_{-\{2j-1, 2j\}} = \vv | Y_{-j} = \yv_j] > 0$, we have that \begin{align*}
    \Pr_{X \sim D}[Y_j = \yv_j \mid Y_{-j} = \yv_{-j} \text{ and } X_{-\{2j-1, 2j\}} = \vv] = e^{\pm (4\delta + 3\eps)} D_{\mathcal{L}, \sqrt{2}s}(\yv_j).
\end{align*}

So take any $j \in [M]$ and $\yv \in \mathcal{L}^M$, and any such $\vv \in (\mathcal{L}')^{N-2}$.  
Since $X_{-\{2j-1, -2j\}}$ determines all the entries of $Y_{-j}$, we have that $\Pr[X_{-\{2j-1, 2j\}} = \vv | Y_{-j} = \yv_j] > 0$ implies that $\Pr[Y_{-j} =\yv_{-j} \mid X_{-\{2j-1, 2j\}} = \vv] = 1$. 
Hence, \begin{align*}
    &\, \Pr_{X \sim D}[Y_j = \yv_j \mid Y_{-j} = \yv_{-j} \text{ and } X_{-\{2j-1, 2j\}} = \vv] \\
    =&\, \Pr_{X \sim D}[Y_j = \yv_j \mid  X_{-\{2j-1, 2j\}} = \vv].
\end{align*}  

Writing $\cv := X_{2j-1} \bmod{\mathcal{L}} = X_{2j} \bmod{\mathcal{L}}$, we have that {\allowdisplaybreaks \begin{align*}
        &\, \Pr_{X \sim D}[Y_j = \yv_j \mid X_{-\{2j-1, 2j\}} = \vv] \\
        =&\,  \Pr_{X \sim D}[X_{2j-1} - X_{2j} = \yv_j \mid X_{-\{2j-1, 2j\}} = \vv] \\
        =&\, \sum_{\xv \in \mathcal{L} + \cv} \Pr_{X \sim D}[X_{2j-1} = \xv \text{ and } X_{2j} = \xv - \yv_j \mid X_{-\{2j-1, 2j\}} = \vv] \\ 
        %=&\, \sum_{\xv \in \mathcal{L} + \cv} \Pr_{X \sim D}[X_{2j-1} = \xv \mid X_{-(2j-1)} \equiv (\vv,\xv - \yv_j)] \cdot \Pr_{X \sim D}[X_{2j} = \xv - \yv_j \mid X_{-\{2j-1, 2j\}} = \vv] \\ % NB: the use of \equiv is not defined here
        =&\, e^{\pm 2\delta} \cdot \sum_{\xv \in \mathcal{L} + \cv}  D_{\mathcal{L} + \cv, s}(\xv) \cdot \Pr_{X \sim D}[X_{2j} = \xv - \yv_j \mid X_{-\{2j-1, 2j\}} = \vv] 
    \end{align*} 
    where the last line follows from the definition of conditional probability and from applying \Cref{eq: output distribution bucket and combine conditioned on cosets} to $i = 2j-1$.} 
    Now,  {\allowdisplaybreaks \begin{align*}
        &\, \Pr_{X \sim D}[X_{2j} = \xv - \yv_j \mid X_{-\{2j-1, 2j\}} = \vv] \\
        =&\, \sum_{\zv \in \mathcal{L} + \cv} \Pr_{X \sim D}[X_{2j-1} = \zv \text{ and } X_{2j} = \xv - \yv_j \mid X_{-\{2j-1, 2j\}} = \vv] \\ 
        %=&\, \sum_{\zv \in \mathcal{L} + \cv} \Pr_{X \sim D}[X_{2j} = \xv - \yv_j \mid X_{-2j} \equiv (\vv, \zv)] \cdot \Pr_{X \sim D}[X_{2j-1} = \zv \mid X_{-\{2j-1, 2j\}} = \vv] \\
        =&\, e^{\pm 2\delta} \cdot D_{\mathcal{L} + \cv, s}(\xv - \yv_j)  \cdot \sum_{\zv \in \mathcal{L} + \cv}   \Pr_{X \sim D}[X_{2j-1} = \zv \mid X_{-\{2j-1, 2j\}} = \vv] \\ 
        %=&\, e^{\pm 2\delta} \cdot D_{\mathcal{L} + \cv, s}(\xv - \yv_j)  \cdot    \frac{\sum_{\zv \in \mathcal{L} + \cv} \Pr[X_{2j-1} = \zv \text{ and } X_{-\{2j-1, 2j\}} = \vv]}{\Pr[X_{-\{2j-1, 2j\}} = \vv]} \\ 
        =&\, e^{\pm 2\delta} \cdot D_{\mathcal{L} + \cv, s}(\xv - \yv_j)
    \end{align*}
    where the second line again follows from the definition of conditional probability and \Cref{eq: output distribution bucket and combine conditioned on cosets}.}
    Therefore, it follows that \begin{align*}
        \Pr_{X \sim D}[Y_j = \yv_j \mid X_{-\{2j-1, 2j\}} = \vv] &= e^{\pm 4\delta} \cdot \sum_{\xv \in \mathcal{L} + \cv}  D_{\mathcal{L} + \cv, s}(\xv) \cdot D_{\mathcal{L} + \cv, s}(\xv - \yv_j) \\
        &= e^{\pm 4\delta} \cdot \Pr_{(X_1, X_2) \sim (D_{\mathcal{L} + \cv, s})^2}[ X_1 - X_2 = \yv_j] \\
        &= e^{\pm (4\delta + 3\eps)} \cdot D_{\mathcal{L}, \sqrt{2}s}(\yv_j)
    \end{align*}
    by \Cref{lem: convolution lemma with specified distance} and since $\frac{1 + \eps}{1 - \eps} \leq e^{3\eps}$ for $\eps \leq \frac{1}{2}$. This completes the proof. 
\qed 
\end{proof}

\subsection{The Algorithm: Wagner as a Gaussian sampler}\label{sec: Wagner as Gaussian sampler algorithm} 

We are now ready to present our Gaussian sampler, laid out as~\Cref{alg: Wagner-Style Gaussian Sampler}. 

\vspace{-0.4cm}

\begin{algorithm}[ht!]
\caption{Wagner-Style Gaussian Sampler}
\label{alg: Wagner-Style Gaussian Sampler}
\DontPrintSemicolon
\SetKwInOut{Input}{Input}\SetKwInOut{Output}{Output}

\Input{Integers $n,m,q$; \\ 
Full-rank matrix $\Am = [\Am' \mid \Id_n] \in \Z_q^{n \times m}$; \\
Integer parameters $N, r, (p_i)_{i=1}^r, (\nnew_i)_{i=1}^r$ with $\sum_{i=1}^r \nnew_i = n$; \\
Real parameter $s_0 > 0$}
\Output{List of vectors in $\Lambda_q^\bot(\Am)$}
\vspace{2mm}
Let $\Lambda_0 := \Z^{m-n}$ \; 
Initialize a list $L_0$ with $3^r N$ independent samples from $D_{\Lambda_0, s_0}$  \; 
\For{$i = 1, \ldots, r$}{
    Let $\Lambda_i$ as defined in \Cref{eq: def of Lambda} \;
    Let $\Lambda'_i = \mathcal{L}(\Bm'_i)$ for $\Bm'_i$ as defined in \Cref{eq: basis B'i} \;
    $L'_{i-1} := \emptyset$ \; 
    \For(\tcp*[f]{Lift to $\Lambda'_i$ \hspace{0.3cm}}){$\xv \in L_{i-1}$}{
    Sample $\xv' \sim \mathrm{DGLift}(\Lambda_{i-1}, \Lambda'_i, s_{i-1}, \xv)$ $\>\vartriangleright$ \Cref{alg: DGLift} \; 
    Append $\xv'$ to $L'_{i-1}$}
    $L_i = \mathrm{BucketAndCombine}(\Lambda'_i, \Lambda_i, L'_{i-1})$ $\>\vartriangleright$ \Cref{alg: bucket and combine} \tcp*{Combine to $\Lambda_i$}  
    $s_i := \sqrt{2} s_{i-1}$ \; 
}
\Return{$L_r$} 
\end{algorithm} % NB: Actually DGLift has embedding of \Lambda_{i-1} as input 
\vspace{-0.4cm}

\begin{remark}
    In our applications of \Cref{alg: Wagner-Style Gaussian Sampler}, we consider input parameter $s_0$ satisfying $s_0 \geq \sqrt{\ln(2(m-n) + 4)/\pi}$  and  $\sqrt{2^{i-1}} s_0\geq \frac{q}{p_i} \sqrt{\ln(2\nnew_i + 4)/\pi}$ for all $i \in [r]$. 
    This allows us to use the exact sampler from \Cref{lem: exact randomized Gaussian rounding} to sample from $D_{\Lambda_0, s_0}$ (in the first iteration) and from $D_{\frac{q}{p_i} \Z^{\nnew_i}, \sqrt{2^{i-1}} s_{0}}$ (in iterations $i = 1,\ldots, r$). % namely in the application of \Cref{alg: DGLift} in iteration $i \in [r]$  
\end{remark}

\begin{theorem}[Correctness of One Iteration]\label{thm: correctness of one iteration Gaussian Wagner algorithm}
    Let $\delta \geq 0$ and $0 < \eps \leq \frac{1}{2}$ be reals. For $i \in [r]$, consider iteration $i$ of \Cref{alg: Wagner-Style Gaussian Sampler} with well-defined parameters.
    If $s_{i-1} \geq \max(\eta_\eps(\frac{q}{p_i}\Z^{\nnew_i}), \sqrt{2}\eta_\eps(\Lambda'_i), \frac{q}{p_i} \sqrt{\ln(2\nnew_i + 4)/\pi})$ % NB: It also suffices to have \eps \leq 1/(2\nnew_i + 4) 
    and $L_{i-1}$ consists of $|L_{i-1}| \geq 3 p_i^{\nnew_i}$ vectors in $\Lambda_{i-1}$ that are conditionally $\delta$-similar to independent samples from $D_{\Lambda_{i-1}, s_{i-1}}$, then $L_i$ consists of $\lfloor |L_{i-1}|/3 \rfloor$ vectors in $\Lambda_i$ that are conditionally $(4\delta + 15\eps)$-similar to independent samples from $D_{\Lambda_i, \sqrt{2}s_{i-1}}$.  
\end{theorem}

\begin{proof}  
Let $\mathcal{S}$ be the embedding of $\frac{q}{p_i}\Z^{\nnew_i}$ in $\R^{m - n + \nsum_i}$ and $\mathcal{P}$ the embedding of $\Lambda_{i-1}$ in $\R^{m - n + \nsum_i}$.  Note that $\mathcal{S}$ is a primitive sublattice of $\Lambda'_i$ and that $\mathcal{P} = \pi_{\mathcal{S}}^{\bot}(\mathcal{C})$ for some complement $\mathcal{C}$ to $\mathcal{S}$. 
Since $s_{i-1} \geq \eta_\eps(\frac{q}{p_i}\Z^{\nnew_i}) = \eta_\eps(\mathcal{S})$, the application of the algorithm from \Cref{lem: conditionally similar DGLift} turns the $|L_{i-1}|$ random variables on $\Lambda_{i-1}$ into $|L_{i-1}|$ random variables on $\Lambda'_i$ that are conditionally $(\delta + 3\eps)$-similar to $D_{\Lambda'_i, s_{i-1}}$. 
Here, as oracle to sample from  $D_{\frac{q}{p_i}\Z^{\nnew_i}, s_{i-1}, \cv}$ (exactly),  we use \Cref{lem: exact randomized Gaussian rounding} to sample from $D_{\Z^{\nnew_i}, \frac{p_i}{q} s_{i-1}, \frac{p_i}{q}\cv}$, and multiply the resulting vector by $\frac{q}{p_i}$. 
This is justified since we assume that $\frac{p_i}{q} s_{i-1} \geq \sqrt{\ln(2\nnew_i + 4)/\pi}$.  

By \Cref{lem: distribution of bucket and combine} (where we use that $s_{i-1} \geq \sqrt{2}\eta_\eps(\Lambda'_i)$), the output list $L_i$ of $\mathrm{BucketAndCombine}$ consists of  $\lfloor |L_{i-1}|/3 \rfloor$ vectors in $\Lambda_i$ that are conditionally $\delta'$-similar to independent samples from $D_{\Lambda_i, \sqrt{2}s_{i-1}}$, where $\delta' = 4(\delta + 3\eps) + 3\eps = 4\delta + 15\eps$. 
\qed 
\end{proof}

\begin{theorem}[Correctness of \Cref{alg: Wagner-Style Gaussian Sampler}]\label{thm: correctness of full Gaussian Wagner algorithm} 
    Let $0 < \eps \leq \frac{1}{2}$ be a real, and $r$ an integer. 
    If the input parameters satisfy $N \geq p_i^{\nnew_i}$, $s_0 \geq \sqrt{\ln(2(m-n) + 4)/\pi}$, and $\sqrt{2^{i-1}} s_0 \geq \max(\eta_\eps(\frac{q}{p_i}\Z^{\nnew_i}), \sqrt{2}\eta_\eps(\Lambda'_i), \frac{q}{p_i} \sqrt{\ln(2\nnew_i + 4)/\pi})$ for all $i \in [r]$, then Algorithm~\ref{alg: Wagner-Style Gaussian Sampler} returns a list of size $N$ consisting of vectors that are conditionally $4^r 5 \eps$-similar to independent samples from $D_{\Lambda_q^\bot(\Am), \sqrt{2^r} s_0}$.
\end{theorem}

\begin{proof}
To obtain the list $L_0$, we use the exact $D_{\Z^{m-n}, s_0}$-sampler from \Cref{lem: exact randomized Gaussian rounding}, which takes time $\mathrm{poly}(m-n, \log s_0)$. This is justified since we assume $s_0 \geq \sqrt{\ln(2(m-n) + 4)/\pi}$. 

We show by induction that, for each iteration $i \in \{1,\ldots,r\}$, the output list $L_i$ consists of $3^{r-i}N$ vectors in $\Lambda_i$ that are conditionally $\delta_i$-similar to independent samples from $D_{\Lambda_i, s_{i}}$ for $\delta_i := (4^i - 1) 5 \eps$. The theorem then immediately follows, since $(4^r - 1)5\eps \leq 4^r 5 \eps$. 

By assumption, $|L_0| = 3^r N \geq 3^r p_1^{\nnew_1} \geq 3 p_1^{\nnew_1}$ and $s_0 \geq \max(\eta_\eps(\frac{q}{p_1}\Z^{\nnew_1}), \sqrt{2}\eta_\eps(\Lambda'_1), \linebreak[0] \frac{q}{p_1} \sqrt{\ln(2\nnew_1 + 4)/\pi})$. % NB: Linebreak
Therefore, \Cref{thm: correctness of one iteration Gaussian Wagner algorithm} implies that the output list $L_1$ of iteration $i=1$ consists of $3^{r-1} N$ vectors in $\Lambda_1$ that are conditionally $15\eps$-similar to independent samples from $D_{\Lambda_1, s_{1}}$. Since $\delta_1 = 15\eps$, this proves the base case. 

Consider any $i \in \{2,\ldots, r\}$, and suppose the claim holds for all $1 \leq j \leq i-1$. By the induction hypothesis and assumption, $|L_{i-1}| = 3^{r-(i-1)} N \geq 3^{r-(i-1)} p_i^{\nnew_i} \geq 3p_i^{\nnew_i}$ and $s_{i-1} = \sqrt{2^{i-1}} s_0 \geq \max(\eta_\eps(\frac{q}{p_i}\Z^{\nnew_i}), \sqrt{2}\eta_\eps(\Lambda'_i),  \frac{q}{p_i} \sqrt{\ln(2\nnew_i + 4)/\pi})$. Therefore, \Cref{thm: correctness of one iteration Gaussian Wagner algorithm} implies that the output list $L_i$ of iteration $i$ consists of $3^{r-i} N$ vectors in $\Lambda_i$. By the induction hypothesis the vectors in $L_{i-1}$ are conditionally $\delta_{i-1}$-similar to independent samples from $D_{\Lambda_{i-1}, s_{i-1}}$, so \Cref{thm: correctness of one iteration Gaussian Wagner algorithm} yields that the vectors in $L_i$ are conditionally $(4\delta_{i-1} + 15\eps)$-similar to independent samples from $D_{\Lambda_i, s_{i}}$. Since 
$4\delta_{i-1} + 15\eps = 4(4^{i-1} - 1)5\eps + 15\eps = (4^i - 1)5\eps = \delta_i$, the claim follows. 
\qed 
\end{proof}

\begin{remark}[Expected Runtime]\label{rem: exp runtime of full alg}   
Using the exact sampler from \Cref{lem: exact randomized Gaussian rounding} ensures that all vectors processed by \Cref{alg: Wagner-Style Gaussian Sampler} have expected bitsize at most  $\mathrm{poly}(m,r,\log s_0, \log q)$ when the parameters satisfy $p_i \leq q$ for all $i \in [r]$. % \log q for DGLift; vectors may be added at most r times  
Hence, the expected runtime of \Cref{alg: Wagner-Style Gaussian Sampler} is then at most $(3^rN + \sum_{i=1}^{r}|L_{i-1}|)  \cdot \mathrm{poly}(m,r,\log s_0, \log q) = 3^r N \cdot \mathrm{poly}(m,r, \log s_0, \log q)$. 
%(Here, we use that the application of \Cref{lem: exact randomized Gaussian rounding} in the application of  $\mathrm{DGLift}$  during iteration $i$ has runtime $\mathrm{poly}(\nnew_i,\frac{p_i}{q} s_{i-1})$ and that multiplication by $\frac{q}{p_i}$ can be done in time $\poly(\log q)$.) 
\end{remark}

\subsection{Putting It All Together} \label{sec: complexity analysis wrap-up} \label{sec: complexity analysis}  

We now have all the ingredients to prove the existence of a subexponential-time algorithm for sampling from a distribution conditionally similar to $D_{\Lambda_q^\bot(\Am), s}$ for random (full-rank) matrices $\Am$. 
We write the width $s$ of the desired discrete Gaussian distribution as $s = q/f$ for some $f>1$, and remark that the difficulty of sampling with width $q/f$ increases with $f$. Below, we demonstrate that subexponential complexity is feasible for all $q/f$ above a certain threshold.  

We note that, for the parameters of interest, we have $q^{1-n/m} \geq 2^{\Theta(\log n / \log \log n)}$, which tends to infinity as $n$ grows (so the assumption $q^{1-n/m} \geq 6$ is not too restrictive). %= q^{\Theta(1/ \log \log n)}  

\begin{theorem}\label{thm: existence of subexponential Wagner sampler} 
     For $n \in \N$, let $m \geq n$ be an integer and $q = \mathrm{poly}(n)$ be a prime such that $q^{1-n/m} \geq 6$. 
     Let $f>1$ and $\eps \leq \frac{1}{m}$ be positive reals such that $\frac{q}{f} \geq \sqrt{\ln(1/\eps)}$. 
     For sufficiently large $n$, there exists a randomized algorithm that, given a uniformly random full-rank matrix $\Am \in \Z_q^{n \times m}$, with probability $> 1 - 2^{-\Tilde{\Omega}(n)}$ returns $N$ vectors in $\Lambda_q^\bot(\Am)$, where $N \in 2^{o(m-n)}$ is such that
     \begin{align}\label{eq: bound on log N}
         \log_2(N) = \frac{n/2}{\ln(\ln(q)) - \ln\left(\ln(f) + \frac{1}{2} \ln\ln(1/\eps)\right) - O(1)}.
     \end{align} 
     The output vectors are conditionally $q^4\eps$-similar to independent samples from $D_{\Lambda_q^\bot(\Am), \frac{q}{f}}$.
     This algorithm has time and memory complexity $N \cdot \poly(m)$. 
\end{theorem}

\begin{proof} 
We first provide a choice of parameters for \Cref{alg: Wagner-Style Gaussian Sampler} (instantiated with the exact sampler from \Cref{lem: exact randomized Gaussian rounding}). Then, we show that they satisfy the conditions of \Cref{thm: correctness of full Gaussian Wagner algorithm}, and conclude the proof.

\textbf{Choice of Parameters.}
Define $\eps' := \eps/5$, and let $N$ be the smallest possible integer such that 
\begin{align*}
    \log_2(N/q) \geq \frac{n/2}{\ln\ln(q) - \ln\left[(\ln(f) + \frac{1}{2}\ln(\frac{144}{\pi} \ln(3/\eps')) + 1/2 \right]},
\end{align*}
which satisfies \Cref{eq: bound on log N} for sufficiently large $n$. Let $s_0 := \frac{q/f}{\sqrt{2^r}}$, where 
\footnote{For all interesting parameters, we have $r \geq 1$. For the rare setting of parameters for which $r < 1$, we replace $r$ by $r+10$. (This suffices since, for any valid parameters, the assumption $\frac{q}{f} \geq \sqrt{\ln(1/\eps)}$ ensures that $2\log_2(q/f) - \log_2(144 \ln(3/\eps')/\pi) + 10 \geq 1$ when $n \geq 2$.) Note that increasing $r$ by a constant additive factor does not affect the proof; in particular, Equation~\eqref{eq: wrap up proof, decreases with r} below would still hold as it decreases with $r$.} % Reason: $\tfrac{144 \ln(3/\eps')}{\pi} = \tfrac{144 \ln(15/\eps)}{\pi}\leq e^6 \ln(1/\eps)$ for $1/\eps \geq 2$. Hence, $2\log_2(q/f) - \log_2(144 \ln(3/\eps')/\pi) \geq \log_2(\ln(1/\eps)) -  \log_2(144 \ln(3/\eps')/\pi) \geq - \log_2(e^6) \geq -9$.
$$r := \lfloor 2\log_2(q/f) - \log_2(\tfrac{144 \ln(3/\eps')}{\pi}) \rfloor.$$
For $i \in [r]$, define $p_i := \lfloor q/\sqrt{2^i} \rfloor$. For $i \in [r-1]$, define $\nnew_i := \lceil \frac{\log_2(N)}{\log_2(q) - i/2} - 1 \rceil$, and define $\nnew_r := n - \sum_{i=1}^{r-1} \nnew_i$ so that $\sum_{i=1}^r \nnew_i = n$. % Note that this implies $s_0 \geq \sqrt{\frac{144 \ln(3/\eps')}{\pi}}$. 
Note that $p_i, \nnew_i$ are integers.

We now show that these parameters satisfy the conditions of \Cref{thm: correctness of full Gaussian Wagner algorithm}. We consider sufficiently large $n$ so that $\eps' \leq \frac{1}{6}$ (i.e., $\eps \leq \frac{5}{6}$) and $2\log_2(N) \leq m-n$ (justified by the condition $\log_2(N) = o(m-n)$). 

\textbf{Verifying that $N$ is Large Enough.}
We claim that $N \geq p_i^{\nnew_i}$ for all $i \in [r]$. For all $i \in [r-1]$, this claim follows immediately from the definition of $\nnew_i$: $\nnew_i \leq \frac{\log_2(N)}{\log_2(q) - i/2} \leq \frac{\log_2(N)}{\log_2(p_i)}$. 
So it remains to show that $b_r \leq \frac{\log_2(N)}{\log_2(p_r)}$. %by our choice of $N$ 
By definition, \begin{align}
    \log_2(N/q) &\geq \frac{n/2}{\ln\ln(q) - \ln\left[(\ln(f) + \frac{1}{2}\ln(\frac{144}{\pi} \ln(3/\eps')) + 1/2 \right]} \notag \\
    &\geq \frac{n/2}{\ln(\log_2(q)) - \ln\left[\log_2(f) + \frac{1}{2}\log_2(\tfrac{144 \ln(3/\eps')}{\pi}) + 1/2\right]} \notag \\
    &\geq \frac{n/2}{\ln\left( \frac{2\log_2(q)}{2\log_2(q) - r} \right)} \label{eq: wrap up proof, decreases with r}
\end{align}
since $r \geq  2\log_2(q/f) - \log_2(\tfrac{144 \ln(3/\eps')}{\pi}) - 1$. 
In particular, using known facts of integrals (recall \Cref{footnote: integral facts} on page \pageref{footnote: integral facts}) we obtain \begin{align*}
    n \leq 2\log_2(N/q) \cdot \ln\left( \frac{2\log_2 q }{2\log_2 q -r} \right) = \int_0^{r} \frac{2\log_2(N/q)}{2\log_2 (q) - x} dx \leq \sum_{i=1}^r \frac{2\log_2(N/q)}{2\log_2(q) - i}.
\end{align*}
Since $\frac{2\log_2(N/q)}{2\log_2(q) - i}  \leq \frac{2\log_2(N)}{2\log_2(q) - i} - 1 \leq \nnew_i$, we obtain that $n \leq \sum_{i=1}^{r-1} \nnew_i + \frac{2\log_2(N/q)}{2\log_2(q) - r}$, and thus $\nnew_r = n - \sum_{i=1}^{r-1} \nnew_i  \leq \frac{2\log_2(N/q)}{2\log_2(q) - r} \leq \frac{2\log_2(N)}{2\log_2(q) - r} \leq \frac{\log_2(N)}{\log_2(p_r)}$, as desired.

\textbf{Verifying the Smoothing Conditions of \Cref{thm: correctness of full Gaussian Wagner algorithm}.}
To show that the smoothing conditions in \Cref{thm: correctness of full Gaussian Wagner algorithm} are satisfied with probability $1 - 2^{- \Tilde{\Omega}(n)}$, it suffices to show that the following holds (for large enough $n$): \begin{outline}
    \1[(I) ] $s_0 \geq \sqrt{\ln(2(m-n) + 4)/\pi}$ and $\sqrt{2^{i-1}} s_0 \geq \frac{q}{p_i} \sqrt{\ln(2\nnew_i + 4)/\pi}$ for all $i \in [r]$.
    \1[(II)] With probability $1 - 2^{- \Tilde{\Omega}(n)}$, we have that  $\sqrt{2^{i-1}} s_0 \geq \sqrt{2} \max( \eta_{\eps'/3}(\frac{q}{p_i}\Z^{\nnew_i}), \linebreak[0] \eta_{\eps'/3}(\Lambda_{i-1}))$ for all $i \in [r]$. % NB: Linebreak
\end{outline}
Indeed, if $\sqrt{2^{i-1}} s_0 \geq \sqrt{2} \max( \eta_{\eps'/3}(\frac{q}{p_i}\Z^{\nnew_i}), \eta_{\eps'/3}(\Lambda_{i-1}))$ (for any $i\in [r]$), then $\sqrt{2^{i-1}} s_0 \geq \max(\eta_{\eps'}(\frac{q}{p_i}\Z^{\nnew_i}), \sqrt{2}\eta_{\eps'}(\Lambda'_i))$ by~\cite[Proposition~2]{EWY23} %(applied with $\mathcal{S}$ the embedding of $\frac{q}{p_i}\Z^{\nnew_i}$ in $\R^{m - n + \nsum_i}$ and $\mathcal{P}$ the embedding of $\Lambda_{i-1}$ in $\R^{m - n + \nsum_i}$) 
(with the embedding of $\frac{q}{p_i}\Z^{\nnew_i}$ in $\R^{m - n + \nsum_i}$ as sublattice) and because $\eta_{{\eps'}/3}(\frac{q}{p_i}\Z^{\nnew_i}) \geq \eta_{\eps'}(\frac{q}{p_i}\Z^{\nnew_i})$.

To prove (I) and (II), we will use that our choice of $r$ implies that $s_0 \geq \sqrt{\frac{144 \ln(3/\eps')}{\pi}}$. 
Furthermore, we emphasize that $\eps \leq \frac{1}{m}$ implies $\eps' \leq \frac{3}{4m}$, so $\eps' \leq \min(\frac{3}{4(m-n)}, \frac{3}{4\nnew_1})$ and  $\eps' \leq \min(\frac{3}{4(m - n + \nsum_{i-1})}, \frac{3}{4\nnew_i})$ for all $i \in \{2, \ldots, r\}$. (This allows us to apply the results from \Cref{sec: relevant smoothing bounds} to bound the parameters $\eta_{\eps'}(\frac{q}{p_i}\Z^{\nnew_i})$ and $\eta_{\eps'}(\Lambda'_i)$.)

So consider any $i \in [r]$. By \Cref{lem: lemma 3.3 in MR2004 applied to Zn} (with dimension $\nnew_i$ and using that ${\eps'} \leq \frac{3}{4\nnew_i}$), we obtain $\sqrt{2}\eta_{{\eps'}/3}(\frac{q}{p_i} \Z^{\nnew_i}) \leq \frac{q}{p_i} \sqrt{\frac{4\ln(3/{\eps'})}{\pi}} \leq \sqrt{2^{i-1}} \sqrt{\frac{32\ln(3/{\eps'})}{\pi}} \leq \sqrt{2^{i-1}} s_0$ (since $p_i = \lfloor q/\sqrt{2^i} \rfloor \geq (q/\sqrt{2^i})/2$ whenever $p_i \geq 2$, and thus $\frac{q}{p_i} \leq 2\sqrt{2^i}$). 
Furthermore, for $i = 1$, we have by \Cref{lem: lemma 3.3 in MR2004 applied to Zn} that $\sqrt{2}\eta_{{\eps'}/3}(\Lambda_0) = \sqrt{2}\eta_{{\eps'}/3}(\Z^{m-n}) < \sqrt{\frac{4 \ln(3/{\eps'})}{\pi}} \leq s_0$. % (where we note that ${\eps'} \leq \frac{1}{4(m-n)}$). 
Since ${\eps'} \leq \frac{3}{4m}$ we have that $\frac{3}{\eps'} \geq 2m \geq 2(m-n) + 4$ and $\frac{3}{\eps'} \geq 4m \geq 4n \geq 2n + 4 \geq 2\nnew_i + 4$ for all $i\in[r]$. Hence, $s_0 \geq \sqrt{\frac{4 \ln(3/{\eps'})}{\pi}}$ implies that $s_0 \geq \sqrt{\frac{\ln(3/{\eps'})}{\pi}}\geq \sqrt{\frac{ \ln(2(m-n) + 4)}{\pi}}$, and $\sqrt{2^{i-1}} s_0 \geq \frac{q}{p_i} \sqrt{\frac{4\ln(3/{\eps'})}{\pi}}$ implies that $\sqrt{2^{i-1}} s_0 \geq \frac{q}{p_i} \sqrt{\frac{\ln(3/{\eps'})}{\pi}} \geq  \frac{q}{p_i} \sqrt{\frac{\ln(2\nnew_i + 4)}{\pi}}$ for all $i \in [r]$.
Thus, (I) holds for $n \geq 2$. 

Note that $q^{\frac{\nsum_j}{m-n + \nsum_j}} \leq \sqrt{2^j}$ for all $j \in [r]$. To see this, observe that (for any  $j \in [r]$) 
$\nsum_{j} \leq j  \nnew_j$, so that $\frac{\nsum_{j}}{j} \left(2 \log_2(q) - j \right) \leq \nnew_j (2 \log_2(q) - j) \leq 2\log_2(N) \leq m - n$. Since $\frac{\nsum_{j}}{j} \left(2 \log_2(q) - j \right) \leq m - n$ if and only if $q^{\frac{\nsum_j}{m-n + \nsum_j}} \leq \sqrt{2^j}$, the claim follows. 
Thus, for each $i \in \{2,\ldots, r\}$, \Cref{lem: smoothing parameter of q-ary kernel lattice} and the previous claim imply that $$\sqrt{2}\eta_{{\eps'}/3}(\Lambda_{i-1}) < \sqrt{\frac{144 \ln(3/{\eps'})}{\pi}} \cdot q^{\frac{\nsum_{i-1}}{m-n + \nsum_{i-1}}} \leq \sqrt{\frac{144 \ln(3/{\eps'})}{\pi}} \sqrt{2^{i-1}} \leq \sqrt{2^{i-1}} s_0,$$ except with probability $< 2^{-\nsum_{i-1}}$. % (Note that $m \geq n$ implies that $m - n + \nsum_i \geq \nsum_i$. In addition, $n \geq \nsum_i$ ensures that $n(m-n)\geq \nsum_i(m-n)$ and thus $\frac{n_i}{m-n+n_i} \leq \frac{n}{m}$, which implies $q^{1 - \frac{n_i}{m-n+n_i}} \geq q^{1 - \frac{n}{m}} \geq 6$. Finally, ${\eps'} \leq \frac{1}{4(m - n + \nsum_i)}$.) 

Therefore, by the union bound, we have with probability $ > 1 - \sum_{i=1}^{r-1} 2^{-\nsum_i} $ that $\sqrt{2^{i-1}} s_0 \geq \sqrt{2} \max( \eta_{\eps'/3}(\frac{q}{p_i}\Z^{\nnew_i}), \linebreak[0] \eta_{\eps'/3}(\Lambda_{i-1}))$  for all $i \in [r]$.
Note that $1 - \sum_{i=1}^{r-1} 2^{-\nsum_i} \geq 1 - r 2^{-\nnew_1} = 1 - 2^{-(\nnew_1 + \log_2(r))} = 1 - 2^{- \Tilde{\Omega}(n)}$  since $\log_2(r) = \log_2\log_2(n) + O(1)$, $\nnew_1 = \Omega(\frac{\log_2 N}{\log_2 n})$, and $\log_2(N) = \Omega(\frac{n}{\log_2\log_2(n)})$, so (II) holds as well.\footnote{We emphasize that it also follows that $\eta_{\eps/4}(\Lambda_q^\bot(\Am)) < \frac{q}{f}$. Indeed, another application of \Cref{lem: smoothing parameter of q-ary kernel lattice} (using that $\eps/4 \leq \frac{1}{4m}$) yields $\eta_{\eps/4}(\Lambda_q^\bot(\Am)) < \sqrt{72  \ln(4/\eps)/\pi} q^{n/m}$, except with probability $< 2^{-n}$, but this does not affect the lower bound on the success probability. By the aforementioned fact, we have that $q^{n/m} \leq \sqrt{2^r}$, so by definition of $r$ we obtain that $\eta_{\eps/4}(\Lambda_q^\bot(\Am)) < \frac{q}{f}$. \label{footnote: bounding q and f and smoothing parameter}}

\textbf{Conclusion of the Proof.}
\Cref{thm: correctness of full Gaussian Wagner algorithm} then implies that the output of this algorithm consists of at least $N$ vectors in $\Lambda_q^\bot(\Am)$ that are conditionally $4^r5\eps'$-similar to independent samples from $D_{\Lambda_q^\bot(\Am), \sqrt{2^r} s_0}$. Since $4^r \leq q^4$, we obtain that they are conditionally $\delta$-similar for $\delta = q^4 5\eps' = q^4 \eps$-similar.

Finally, the runtime of \Cref{alg: Wagner-Style Gaussian Sampler} is at most $3^r N \cdot \mathrm{poly}(m,r, \log s_0, \log q)$ by \Cref{rem: exp runtime of full alg}. Since $r = O(\log q)$, $s_0 \leq q$, and $q = \poly(m)$, it follows that the time and memory complexity are both upper bounded by $N \cdot \poly(m)$.
\hfill $\qed$ \end{proof}

\section{Asymptotic Application to Cryptographic Problems} \label{sec:Asymptotic Application to Cryptographic Problems}

We now apply our previous result on Gaussian sampling for SIS lattices to solving various variants of SIS in subexponential time. We also discuss why our result does not directly lead to a provable subexponential-time algorithm for LWE with narrow error distribution.

\subsection{Implications for $\SISinfty$}

The following theorem instantiated with any $m = n + \omega(n/\log \log n)$ such that\footnote{Note that one may always decrease $m$ by ignoring SIS variables, hence the condition $m \leq \poly(n)$ comes with no loss of generality.} $m \leq \poly(n)$, $q = \poly(n)$, $\beta = q/\polylog(n)$, and with a sufficiently small $\eps = 1/\poly(n)$ provides a subexponential-time algorithm for $\SISinfty_{n,m,q,\beta}$.

\begin{theorem}\label{thm: existence of subexponential Wagner alg for SIS inf}
     For $n \in \N$, let $m = n + \omega(n/ \log \log n)$ be integer and $q = \mathrm{poly}(n)$ be prime such that $q^{1-n/m} \geq 6$. 
     Let $f>1$ and $\eps \leq \frac{1}{mq^4}$ be positive reals such that $\frac{q}{f} \geq \sqrt{\ln(1/\eps)}$. 
     For sufficiently large $n$, there exists an algorithm that solves $\SISinfty_{n,m,q,\beta}$ for $\beta := \frac{q}{f}\sqrt{\ln m}$ in expected time 
        $$T = 2^{\tfrac{n/2}{\ln(\ln(q)) - \ln\left(\ln(f) + \frac{1}{2} \ln\ln(1/\eps)\right) - O(1)}} \cdot \poly(m)$$ 
    with success probability $1 - \frac{1}{\Omega(n)}$. 
\end{theorem}

\begin{proof}
    Apply \Cref{thm: existence of subexponential Wagner sampler} with input $n,m,q,f,\eps$ to the  $\SISinfty_{n,m,q,\beta}$ instance $\Am$. 
    With probability $> 1 - 2^{-\Tilde{\Omega}(n)}$, it returns a list of vectors $\xv_1,\ldots,\xv_N$ in $\Lambda_q^\perp(\Am)$ (so they are solutions to $\Am\xv = \zerovec \bmod q$) that are conditionally $q^4\eps$-similar to independent samples from $D_{\Lambda_q^\perp(\Am), s}$ for $s := \frac{q}{f}$. In particular, it follows that the first vector $\xv_1$ follows a distribution $D$ that is $q^4\eps$-similar to $D_{\Lambda_q^\perp(\Am), s}$ (recall \Cref{rem: conditional similarity implies marginal similarity}). 
    By \Cref{footnote: bounding q and f and smoothing parameter} we have $\frac{q}{f} > \eta_{\eps/4}(\Lambda_q^\bot(\Am))$. We note that without loss of generality, we may assume that $\frac{q}{f} > 2\eta_{\eps/4}(\Lambda_q^\bot(\Am))$ by replacing the role of the constant $144$ in the proof of \Cref{thm: existence of subexponential Wagner sampler} by a larger constant. 

    Since $D$ is $q^4\eps$-similar to $D_{\Lambda_q^\perp(\Am), s}$,  \Cref{lem: easy consequence of conditional similarity} (together with the fact that $e^{-x} \geq 1 - x$ for all $x\in \R$) implies that \begin{align*}
        \Pr[\|\xv_1\|_\infty \leq \beta \wedge \xv_1 \neq\zerovec]  &\geq e^{-q^4\eps} \Pr_{X \sim D_{\Lambda_q^\perp(\Am),s}}[\|X\|_\infty \leq \beta \wedge X \neq\zerovec] \\
        &\geq \Pr_{X \sim D_{\Lambda_q^\perp(\Am),s}}[\|X\|_\infty \leq \beta \wedge X \neq\zerovec] - q^4\eps \\
        &\geq \Pr_{X \sim D_{\Lambda_q^\perp(\Am),s}}[\|X\|_\infty \leq \beta \wedge X \neq\zerovec] - \frac{1}{m}.
    \end{align*}

    We will now show that $\Pr_{X \sim D_{\Lambda_q^\perp(\Am),s}}[\|X\|_\infty \leq \beta \wedge X \neq\zerovec] \geq 1 - \frac{2}{m^2} - \frac{1}{2^{m-1}}$, from which it follows that $\xv_1$ is a solution to $\SISinfty$ with probability at least $1 - \frac{2}{m^2} - \frac{1}{2^{m-1}} - \frac{1}{m} \geq 1 - \frac{3}{m} \geq 1 - \frac{3}{n}$ when $n\geq 2$ (where we use that $m\geq n$), thereby proving the theorem.   
    
    It remains to prove our claim. We have that \begin{align*}
        \Pr_{X \sim D_{\Lambda_q^\perp(\Am), s}}[\|X\|_\infty \leq \beta \wedge X\neq\zerovec] &\geq \Pr_{X \sim D_{\Lambda_q^\perp(\Am), s}}[\|X\|_\infty \leq \beta] - \Pr_{X \sim D_{\Lambda_q^\perp(\Am), s}}[X=\zerovec]. 
    \end{align*}
    As $s = \frac{q}{f} > 2\eta_{\eps/4}(\Lambda_q^\bot(\Am))$, we have by \Cref{lem: upper bound on maximum discrete Gaussian probability (min-entropy)} that $\Pr_{X\sim D_{\Lambda_q^\perp(\Am), s}}[X=\zerovec] \leq \frac{1 + \eps/4}{1 - \eps/4} \cdot 2^{-m} \leq \frac{1}{2^{m-1}}$ (since $\frac{1 + x/4}{1 - x/4} \leq 2$ for all $x \in [0,1]$). 
    Also, \Cref{lem: Lemma 2.10 in Ban95} yields that 
    $\Pr_{X \sim D_{\Lambda_q^\perp(\Am), s}}[\|X\|_\infty \leq \beta] > 1 - 2me^{-\pi(\frac{\beta f}{q})^2} \geq 1 - \frac{2}{m^2}$ (where we use that $\beta = \frac{q}{f} \sqrt{\ln m}$). % Lemma applied for R = \beta/s = \beta f/q (note that lemma is stated for Gaussian width 1)
    So our claim follows.
\qed  
\end{proof}

One may remark that our application of~\Cref{lem: Lemma 2.10 in Ban95} makes us lose a $\sqrt{\ln m}$ factor on the norm bound to reach a constant success probability per sample. It would be tempting to only aim for a success probability barely greater than $1/N$ instead, however, the proof would then require $\eps \approx 1/N$, which would make the algorithm exponential. This is admittedly a counterintuitive situation, and plausibly a proof artifact. % Work in progress

\subsection{Implications for $\SIS^\times$ and ISIS in $\ell_2$ norm}

Regarding the $\ell_2$-norm, for the same parameters as above, our Gaussian sampler outputs vectors of length less than $\beta = \sqrt{m} \cdot q/f$ for any $f = \polylog(n)$ in subexponential time. However, this $\SIS$ in $\ell_2$-norm is trivial for such a bound; for example, $(q, 0, \dots, 0)$ is a valid solution. 

Yet, some schemes~\cite{ETWY22} have used the inhomogeneous version of $\SIS$ ($\ISIS$) for bounds $\beta > q$, which was shown~\cite{DEP23} to be equivalent to solving $\SISs$, a variant of $\SIS$ where the solution must be nonzero modulo $q$. The work of~\cite{DEP23} notes that the problem becomes trivial at $\beta \geq q \sqrt{n/12}$ and proposes a heuristic attack that is better than pure lattice reduction when $\beta > q$; however, it appears to run in exponential time in $n$ for $\beta = \sqrt{n} \cdot q/\polylog(n)$. Our Gaussian sampler directly yields a provably subexponential-time algorithm in that regime.

\begin{theorem}\label{thm: existence of subexponential Wagner alg for SIS* 2}
     For $n \in \N$, let $m = n + \omega(n/ \log \log n)$ be integer and $q = \mathrm{poly}(n)$ be prime such that $q^{1-n/m} \geq 6$. 
     Let $f>1$ and $\eps \leq \frac{1}{mq^4}$ be positive reals such that $\frac{q}{f} \geq \sqrt{\ln(1/\eps)}$. 
     For sufficiently large $n$, there exists an algorithm that solves  $\SIS_{n,m,q,\beta}^\times$ and $\ISIS_{n,m,q,\beta}$ for $\beta := \frac{q}{f}\sqrt{m}$ in expected time 
        $$T = 2^{\tfrac{n/2}{\ln(\ln(q)) - \ln\left(\ln(f) + \frac{1}{2} \ln\ln(1/\eps)\right) - O(1)}} \cdot \poly(m)$$ 
    with success probability $1 - \frac{1}{\Omega(n)}$. 
    % NB: Success probability actually appears to be asymptotically larger than in SIS^infty case due to sqrt{m} (instead of \sqrt{\ln m}) 
\end{theorem}

The proof is essentially equivalent to that of~\Cref{thm: existence of subexponential Wagner alg for SIS inf}, up to the invocation of~\Cref{lem: Lemma 2.10 in Ban95} used to tail-bound the norm of a discrete Gaussian, which should be replaced by a similar tail-bound for the $\ell_2$-norm~\cite[Lemma 1.5]{Ban93}. % Applied in dim m with c = \beta/(s \sqrt{m}) and lattice scaled by a factor 1/s. Note that taking \beta = \sqrt{m} q/f (instead of \sqrt{\ln(m)} q/f as in SIS^infty) guarantees that c is not too small / satisfies the premise of the lemma 

Note again here that one may always choose $m = n (1 + o(1))$ without loss of generality even when the given $m$ is much larger, simply by ignoring some SIS variables. Hence, reaching the norm bound $\beta = {q}\sqrt{m} / \polylog(n)$ also permits to reach $\beta = {q}\sqrt{n} / \polylog(n)$.

\subsection{Potential Implications for $\LWE$}

Having obtained a discrete Gaussian sampler for SIS lattices, one may be tempted to apply the dual distinguisher of~\cite{AR05} and directly obtain a subexponential-time algorithm for LWE with narrow secrets. This is in fact problematic because such a distinguisher needs to consider subexponentially many samples simultaneously, so some naive reasoning using the data-processing inequality would force one to instantiate our Gaussian sampler with $\eps = 2^{-\tilde{\Omega}(n)}$ rather than $\eps = n^{-\Theta(1)}$. Unfortunately, for such a small $\eps$ our Gaussian sampler has exponential complexity.  

This issue resonates with the one raised by~\cite{HKM18} regarding the proof of~\cite{KF15} when $m$ is linear in $n$, namely that it is about reaching exponentially small statistical distance. However, it is technically different: in our case, it cannot be fixed by increasing $m$ to $\Theta(n \log n)$. Tracking down the limiting factor leads to blaming the poor smoothing parameters of $\frac q {p_i} \Z^{\nnew_i}$. Thanks to the generality of our framework, it is plausible that this superlattice of $q \Z^{\nnew_i}$ can be replaced by one of similar index with a much better smoothing parameter. We hope that future work will finally be able to provably fix the claim of~\cite{KF15} when $m$ is linear in $n$.

% Print bibliography 
\iftoggle{bibtex}{ 
    \bibliographystyle{alpha}
    \bibliography{references}
}{
	\printbibliography
}

\appendix 

\section{Concrete and Heuristic Application, Including to Dilithium} \label{sec:Concrete and Heuristic Application to Dilithium}

Our asymptotic result for $\SISinfty$ immediately raises the question of its impact on the concrete security of \textsc{Dilithium}. However, the algorithm described above introduces inefficiencies for the sake of provability. There are various details that one would approach differently when aiming to break the problem in practice. It might be that the analysis given below is too aggressive and that the resulting algorithm will fail; in particular, the resulting claim should only be read as a rough underestimation of the cost of the approach. There are certainly further tricks and fine-tuning to be considered, for example from~\cite{BGJMW20}. 

\subsubsection{Relaxing Independence.} If $N$ is the number of buckets for the colliding phase, our provable algorithm used $3^r N$ many initial samples, losing a factor of $3$ on the list size at each iteration, but never re-using a sample twice. If we allow ourselves to re-use a sample in several combinations, then $3N$ samples will be enough to maintain the list size throughout the algorithm. In this case, there will be an average of $3$ samples per bucket $(\xv, \yv, \zv)$, from which we can build $3$ different pairs to be subtracted: $\xv - \yv, \xv - \zv, \zv - \yv$. Note that if the buckets have unequal sizes, the total number of available pairs only increases due to the convexity of the function $x \mapsto {x \choose 2} = \frac {x(x-1)} 2$. % A larger bucket will increase a bucket more than a smaller bucket would 
In the context of BKW, such a heuristic improvement can be traced back to at least Levieil and Fouque~\cite{LF06}.

\subsubsection{Initializing Sparse Ternary Vectors.} Having chosen some $N$, one may set the initial list $L_0$ of vectors from $\Z^{m-n}$ to be as small as possible without duplicates. We therefore choose ternary vectors of $\ell_\infty$-norm $w$, where $w$ is the smallest integer such that $2^w {{m-n} \choose {w}} \geq N$. This gives an initial variance of $\sigma^2_0 = w/( {m-n})$.\footnote{To relate the $\sigma_i$ with the standard deviations of the discrete Gaussian distributions in the rest of the paper, note that a discrete Gaussian distribution with parameter $s_i$ has standard deviation $s_i/\sqrt{2\pi}$.} 

\subsubsection{Rounding and Quantization.} The introduction of Gaussians is also motivated by provability, and one would rather use regular rounding in practice. Following the analysis of~\cite{KF15,GJMS17}, the rounding introduced an error of deviation $\sigma_i = q / (p_i \sqrt{12})$. % NB: Integral of x^2 on [-1/2,1/2] equals 12
Both works also mention that lattice quantization could replace this rounding, which would improve the deviation to $ q / (p_i \sqrt{2 \pi e})$. This allows to choose smaller $p_i$ while maintaining $\sigma_i = 2^{i/2} \cdot \sigma_0$, and therefore increases $\nnew_i$. % Also related: https://www.iacr.org/archive/crypto2015/92160189/92160189.pdf 

\subsubsection{Fractional Parameters.} The parameters $p_i$ and $\nnew_i$ would need to be integers, which might force the attacker not to match exactly the optimal parametrization. % NB: For the theoretical attack our params are integers. In the concrete cost estimation we didn't take this constraint into account for simplicity 
However, the exercise of globally rounding those parameters optimally appears painful. We ignore these constraints for the attacker, leading to further underestimation of the attack cost. 

\subsubsection{Central Gaussian Heuristic.} At step $r$, we have obtained $3 N$ many samples of standard deviation $\sigma_r$ and we wish to know whether one of them is likely to have an $\ell_\infty$-norm bound less than $\beta$. We proceed using the error function as if the distribution at hand was Gaussian; there are many coordinates where this is reasonable, as they result from summing $2^{j+1}$ coordinates for $j \leq r$. This is questionable for the very last coordinates when using rounding; it seems less of an issue when using quantization which should result in a distribution close to uniform over a ball. 

Having computed the success probability $p$ for one sample, we consider the attack successful if $Np > 1/2$.  

\subsubsection{Early Abort.} It could be that after $r' < r$ steps, the remaining dimensions to lift over are nonzero yet small enough that one of the $N$ samples could have small enough coordinates, while running one further step would double the variance of the rest of the vector. Such a trick has already been considered for lattice reduction attacks~\cite{DKLL18,DEP23} on some instances of ISIS and SIS$^\times$, where, for different reasons, one would also naturally have many candidate solutions at hand. 
Hence, we consider such an option when exploring the parameter space. The number of left-over dimensions will be denoted by $\ell = n - \sum_{i=1}^{r'} \nnew_i$.

\subsection{Concrete Analysis Against Dilithium}

We first recall in \Cref{tab:Dilithium_SIS} the $\SISinfty$ parameters underlying the security of  the NIST standard \textsc{Dilithium}~\cite{DKLL18}. 

\begin{table}
    \setlength\tabcolsep{.6em}
    \centering
    \begin{tabular}{c||c|c|c|c|c}
       NIST level  & $n$ & $m$ & $q$ & $\beta$ & $q/\beta$ \\ \hline
        2 & $256 \cdot 4$ & $256 \cdot 9$ & $8380417$ & $350209$ & $23.9$ \\
        3 & $256 \cdot 6$ & $256\cdot 12$ & $8380417$ & $724481$ & $11.6$ \\
        5 & $256 \cdot 8$ & $256\cdot 16$ & $8380417$ & $769537$ & $10.9$
    \end{tabular}
    \caption{$\SISinfty$ parameters underlying the \textsc{Dilithium} scheme.}
    \label{tab:Dilithium_SIS}
\end{table}

Using the script from \url{https://github.com/lducas/Provable-Wagner-SIS}, we evaluate the smallest $N$ such that the attack is successful according to the above analysis. Our script offers both the straightforward rounding and quantization version. The quantization version gave slightly better results, which are given in \Cref{tab:Dilithium_attack}. 

\begin{table}
    \setlength\tabcolsep{.6em}
    \centering
    \begin{tabular}{c||c|c|c|c|c|c}
       NIST level  & $\log_2 N$ & $w$ & $\sigma_0$ & $r'$ & $\sigma_{r'}$ & $\ell$ \\ \hline
       2 & 269.9 & 37 & 0.1700 & 40 & 178277.2 & 28.9 \\
       3 & 343.0 & 47 & 0.1749 & 42 & 366845.5 & 50.5 \\
       5 & 450.2 & 61 & 0.1726 & 42 & 361934.4 & 104.0 \\
    \end{tabular}
    \caption{Best attack parameters against \textsc{Dilithium} according to the heuristic analysis above, and using the quantization trick.}
    \label{tab:Dilithium_attack}
\end{table}

We do not directly provide the cost of the attack in terms of binary gates, but that cost is clearly larger than $N$. Hence, we conclude that Wagner's algorithm alone does not threaten the concrete security of \textsc{Dilithium}; for example at NIST security level 2~\cite{NIST16} (classically, approximately $2^{128}$ hash evaluations), it already requires operating and storing more than $2^{256}$ vectors. % NB: 269 > 256; 2^256 is largest security level ever considered for actual crypto 
% While one may think of tricks to further improve this attack, it would require rather spectacular improvement to be a realistic threat to \textsc{Dilithium}. 

\subsection{When Would Wagner Shine?}

Given the negative concrete application to \textsc{Dilithium}, one may wonder for which parameters this algorithm has a practical chance to actually beat lattice reduction to solve $\SISinfty_{n, m, q, \beta}$, say fixing $\beta = q/4$.   

Given its complexity $2^{O(n/\log \log q)}$, one might be tempted to choose a large $q$, but that is in fact favoring lattice reduction attacks which have a complexity of $2^{O(n \log(n)/\log(q))}$ when $q = n^{\Theta(1)}$ and $m$ is large enough.
So one would rather choose a small $q$ and potentially very large $n$ to make Wagner win. But there is yet another handle, namely the number of variables $m$. Indeed, we proved that Wagner could work for $m$ as small as $n(1 + o(1))$, a regime that is very unfavorable to lattice attacks as it makes the volume of the lattice huge.

For example, consider $\SISinfty$ with $n=500, m=600, q=1000, \beta=q/4$. The security estimation script\footnote{\url{https://github.com/pq-crystals/security-estimates/blob/master/MSIS_security.py}} using the core-SVP methodology gives a cost for lattice reduction attacks of roughly $2^{107}$. Our script for (heuristic and optimistic) Wagner gives $N=2^{54}$. 
We emphasize, however, that both numbers are merely illustrative, and may not be accurate estimates. 
If relevant, one could consider adapting the software of~\cite{WBLW25} to estimate the cost of Wagner for SIS. 

\end{document}

%% file: preamble.tex
\usepackage{parskip}

\usepackage{graphicx} % Required for inserting images
\usepackage{upref}
\usepackage{latexsym}
\usepackage{pdfpages}
\usepackage{xcolor}
\usepackage[hidelinks]{hyperref} % hidelinks to remove color/corners 
\usepackage{enumerate}
\usepackage{outlines}
\usepackage{amsfonts}
\usepackage{varioref}
\usepackage[ansinew]{inputenc}
\usepackage[many]{tcolorbox}
\usepackage[ruled,noend]{algorithm2e} % linesnumbered adds numbers
\usepackage{color,graphics}
\usepackage{comment} 
\usepackage[labelfont=sl]{caption} % [width=0.8\textwidth, justification=centering, font=it]{caption} 
\usepackage{subcaption}
\usepackage{wrapfig}
\usepackage{braket} % Quantum
\usepackage{csquotes}
\usepackage{amssymb}
\usepackage{amsmath}
\usepackage{mathrsfs}
\usepackage{mathtools}
\usepackage{commath}
\usepackage{thmtools,thm-restate}
\usepackage[capitalise,noabbrev]{cleveref}
\usepackage{boxedminipage}
\usepackage{comment}
\usepackage{todonotes}

\counterwithin{theorem}{section}
\counterwithin{lemma}{section}
\counterwithin{corollary}{section}
\counterwithin{definition}{section}
\counterwithin{remark}{section}   
\counterwithin{problem}{section}

\newcommand{\ALLnote}[1]{{\color{cyan} \footnotesize (All: #1)}}
\newcommand{\LDnote}[1]{{\color{red} \footnotesize (Leo: #1)}}
\newcommand{\LD}[1]{{\color{red} \footnotesize (Leo: #1)}}
\newcommand{\LEnote}[1]{{\color{blue} \footnotesize (Lynn: #1)}}
\newcommand{\JLnote}[1]{{\color{orange} \footnotesize (Johanna: #1)}}
\newcommand{\JLLnote}[1]{{\color{orange} \footnotesize #1}}
\newcommand{\DRAFT}[1]{#1}

\newcommand{\DeleteAllComments}{
    \renewcommand{\ALLnote}[1]{}
    \renewcommand{\LDnote}[1]{}
    \renewcommand{\LD}[1]{}
    \renewcommand{\LEnote}[1]{}
    \renewcommand{\JLnote}[1]{}
    \renewcommand{\JLLnote}[1]{}
    \renewcommand{\DRAFT}[1]{}
}

\newcommand{\R}{\mathbb{R}}

\newcommand{\Z}{\mathbb{Z}}
\newcommand{\N}{\mathbb{N}}

\newcommand{\xv}{\mathbf{x}}
\newcommand{\yv}{\mathbf{y}}
\newcommand{\cv}{\mathbf{c}}
\newcommand{\zv}{\mathbf{z}}
\newcommand{\av}{\mathbf{a}}
\newcommand{\bv}{\mathbf{b}}
\newcommand{\sv}{\mathbf{s}}
\newcommand{\tv}{\mathbf{t}}

\newcommand{\vv}{\mathbf{v}}

\newcommand{\zerovec}{\mathbf{0}}

\newcommand{\Id}{\mathbf{I}}
\newcommand{\Am}{\mathbf{A}}
\newcommand{\Bm}{\mathbf{B}}

\newcommand{\Dm}{\mathbf{D}}

\newcommand{\SISs}{\mathrm{SIS}^\times}
\newcommand{\SIS}{\mathrm{SIS}}
\newcommand{\ISIS}{\mathrm{ISIS}}
\newcommand{\SISinfty}{\mathrm{SIS}^\infty}
\newcommand{\LWE}{\mathrm{LWE}}

\newcommand{\poly}{\mathrm{poly}}
\newcommand{\polylog}{\mathrm{polylog}}

\newcommand{\nnew}{b} % \nnew_i is the additional coords considered in iteration i
\newcommand{\nsum}{n} % \nsum_i is now defined as the number of rows of A_i (so the sum of \nnew_j for j <= i
\newcommand{\eps}{\varepsilon}
\newcommand{\Span}{\mathrm{span}}

%% file: tikz/lift.tex
\begin{tikzpicture}[scale=1.2]
    \def\fl{0.2}
   
    % Lambda_i
    \draw (0-0.24,0) node {$\Lambda = \Lambda_r$} ;
    \draw (2,0) node {$\Lambda_{r-1}$} ;
    \draw (4,0) node {$\cdots$} ;
    \draw (6,0) node {$\Lambda_{1}$} ;
    \draw (8.4,0) node {$\Lambda_{0} = \Z^{m-n}$} ;

    % Projections
    \draw[->>, very thick] (0+2*\fl,0)--(2-2*\fl,0) ;
    \draw[->>, very thick] (2+2*\fl,0)--(4-2*\fl,0) ;
    \draw[->>, very thick] (4+2*\fl,0)--(6-2*\fl,0) ;
    \draw[->>, very thick] (6+2*\fl,0)--(8-2*\fl,0) ;
 
    % Lift Lambda'_i
    % \draw () node {Lift};
    
    \draw (6,1) node {$\Lambda'_1$} ; 
    \draw[<<-, very thick] (8-2*\fl,0+\fl)--(6+2*\fl,1);

    \draw (2,1) node {$\Lambda'_{r-1}$} ; 
    \draw[<<-, very thick] (4-2*\fl,0+\fl)--(2+2*\fl,1); 
    % \draw (5 +0.15, 0.8) node[rotate=-25]{lift} ;
    
    \draw (0,1) node {$\Lambda'_r$} ; 
    \draw[<<-, very thick] (2-2*\fl,0+\fl)--(0+2*\fl,1); 
    % \draw (3 +0.15, 0.8) node[rotate=-25]{lift} ;
    
    \draw (4,1) node {$\cdots$} ;
    \draw[<<-, very thick] (6-2*\fl,0+\fl)--(4+2*\fl,1); 
    % \draw (1 +0.15, 0.8) node[rotate=-25]{lift} ;
    
    % Superlattices relations
    \draw (0,0.5) node {\rotatebox[origin=c]{90}{$\subseteq$}} ; 
    \draw (2,0.5) node {\rotatebox[origin=c]{90}{$\subseteq$}} ; 
    \draw (6,0.5) node {\rotatebox[origin=c]{90}{$\subseteq$}} ;

    \draw [red] plot [smooth] coordinates {(8-\fl,\fl) (6-\fl,1+\fl) (6+\fl,-\fl) (5-\fl,.6) };
    \draw [red, dashed] plot [smooth] coordinates {(5-\fl,.6) (4-\fl,1+\fl) (4,-\fl) (3-\fl,.6)};
    \draw [red, ->] plot [smooth] coordinates {(3-\fl,.6) (2-\fl,1+\fl) (2,-\fl) (0-\fl,1+\fl) (0,\fl)};

\end{tikzpicture}

%% file: tikz/matrix.tex
\begin{tikzpicture}[scale=0.75]

    % fill
    \fill[blue!7] (0,-2) rectangle (8,0); % A2
    \fill[blue!14] (0,-1) rectangle (7,0); % A1

    % A 
    \draw (0,0) rectangle (10,-4) ;
    \draw (6,0)--(6,-4) ;
    \draw[<->] (0+0.050,0.2)--(6-0.050,0.2) ;
    \draw (3,0.2) node[above] {$m - n$} ;
    
    % A1
    \draw[dashed] (0,-1)--(7,-1)--(7,0) ;
    \draw[<->] (-0.2,0-0.050)--(-0.2,-1+0.050) ;
    \draw (-0.2, -0.5) node[left] {$\nnew_1$};
    \draw[<->] (6+0.050,0.2)--(7-0.050, 0.2) ;
    \draw (6.5, 0.2) node[above] {$\nnew_1$};
    \draw (1,-0.5) node {$\textbf{A}'_1$} ;
    \draw (6.5,-0.5) node {$\Id_{\nsum_1}$} ;
    % \draw[<->] (6+0.050,1)--(8-0.050, 1) ; \draw (7,1+0.25) node {$\nsum_2$} ; 
    
    % A2
    \draw[dashed] (0,-2)--(8,-2)--(8,0) ;
    \draw[<->] (-0.2,-1-0.050)--(-0.2,-2+0.050) ;
    % \draw[<->] (-1.1,-0.050)--(-1.1,-2+0.050) ; 
    % \draw (-1-0.5, -1) node {$\nsum_2$};
    \draw (-0.2, -1.5) node[left] {$\nnew_2$};
    \draw[<->] (7+0.050,0.2)--(8-0.050, 0.2) ;
    \draw (7.5, 0.2) node[above] {$\nnew_2$};
    \draw (2,-1.1) node {\Large $\textbf{A}'_2$} ;
    \draw (7,-1) node {\large $\Id_{\nsum_2}$} ;
    
    % I
    \draw (8,-2) node {\LARGE $\textbf{I}_n$} ;

    % A'
    \draw (3,-2) node {\LARGE $\textbf{A}'$} ;

\end{tikzpicture}